\documentclass[letterpaper,12pt]{article}
\usepackage[top=1in, bottom=1in, left=1in, right=1in]{geometry}
\usepackage{amssymb, amsmath, amsthm}
\usepackage{mathrsfs, amsfonts, bbm}
\usepackage{enumerate}
\usepackage{graphicx}
\usepackage{chngcntr}

\usepackage{extarrows}
\usepackage[breaklinks=true]{hyperref}
\usepackage{fullpage}
\usepackage{setspace}
\usepackage{xcolor}
\usepackage{comment}
\usepackage{enumerate}
\usepackage{enumitem}
\usepackage{float}
\usepackage{subcaption}
\allowdisplaybreaks[4]
\usepackage{threeparttable}
\usepackage{multirow}
\usepackage{breakcites}

\usepackage[longnamesfirst]{natbib}

\definecolor{darkblue}{rgb}{0,0,.4}
\hypersetup{
colorlinks=true,
citecolor=blue,
urlcolor=blue,
linkcolor=blue,
}

\theoremstyle{plain}
\newtheorem{theorem}{Theorem}
\newtheorem{corollary}{Corollary}
\newtheorem{lemma}{Lemma}
\newtheorem{prop}{Proposition}

\theoremstyle{remark}
\newtheorem{remark}{Remark}
\newtheorem{definition}{Definition}
\newtheorem{example}{Example}

\newcommand{\RN}[1]{
\textup{\uppercase\expandafter{\romannumeral#1}}
}

\DeclareMathOperator*{\argmin}{arg\,min}
\DeclareMathOperator*{\argmax}{arg\,max}

\newcommand{\Var}{\mathrm{Var}}

\begin{document}

\title{Policy Learning under Endogeneity\\Using Instrumental Variables\thanks{I am indebted to Hiroaki Kaido for his guidance and encouragement throughout this project. I am grateful to Iván Fernández-Val and Jean-Jacques Forneron for their suggestions and support. For helpful comments and discussions, I thank Shuowen Chen, Bryan Graham, Sukjin Han, Eric Hardy, Toru Kitagawa, Ming Li, Ismael Mourifi\'e, Yuya Sasaki, Xiaoxia Shi, Kohei Yata, 
and participants at the BU Econometrics Seminar, NASMES 2022, IAAE 2022 Annual Conference, ESAM2022, SETA2022, YES 2022, and MEG 2022. The financial support of the IAAE Student Travel Grant is gratefully acknowledged. All errors are my own.}}
\author{Yan Liu\footnote{Institute of Economic Research, Kyoto University. Email: liuyan@kier.kyoto-u.ac.jp.}}
\date{January 14, 2026}
\maketitle

\begin{abstract}
I propose a framework for learning individualized policy rules in observational data settings characterized by endogenous treatment selection and the availability of an instrumental variable. I introduce encouragement rules that manipulate the instrument. By incorporating the marginal treatment effect (MTE) as a policy invariant parameter, I establish the identification of the social welfare criterion for the optimal encouragement rule. Focusing on binary encouragement rules, I propose to estimate the optimal encouragement rule via the Empirical Welfare Maximization (EWM) method and derive the welfare loss convergence rate. I apply my method to advise on the optimal tuition subsidy assignment in Indonesia.
\end{abstract}

\section{Introduction}
\label{intro}

Policy effects can be heterogeneous, so an important goal for policymakers is to individualize policy interventions to improve social welfare, that is, to assign individuals to different policies based on their observable characteristics. In many cases, policy decisions need to be informed by observational studies or randomized experiments with imperfect compliance, where people endogenously select into treatment. This paper proposes a framework to learn individualized policy interventions in such settings when an instrumental variable (IV) for the treatment is available. Instead of infeasible mandatory treatment rules, I consider a different class of policies that manipulate the instrument, which I refer to as \emph{encouragement rules}. For example, it is highly costly or impossible to force people to (or not to) attend school. A more realistic scenario is to provide a scholarship or a tuition subsidy, whereas the tuition fee is commonly used as an instrument for school attendance. This paper then studies the identification and estimation of optimal encouragement rules.

To identify optimal encouragement rules, I incorporate the \emph{marginal treatment effect} (MTE) framework of \citet{heckman2005structural} to explicitly model the selection into treatment. In this framework, the MTE plays the role of a policy invariant parameter that aids in producing policy counterfactuals. The first main result of this paper is to establish the identification of the social welfare criterion of encouragement rules, which is defined as the average counterfactual outcome, via the identification of the MTE. In this sense, I bridge the literatures on the MTE and on statistical decision rules of policy interventions. More specifically, the social welfare criterion can be represented as a function of the MTE and the propensity score. This social welfare representation has two benefits. First, it helps the policymaker understand how the optimal encouragement rule is driven by heterogeneity in treatment take-up and treatment effects. Second, it suggests a natural way of performing extrapolation. An encouragement rule can induce variation in the propensity score beyond the observed support. This variation represents individuals whose treatment choice is affected not by the observed instrument but by the encouragement rule. Learning about their average outcome requires extrapolation, and point identification can be restored by assuming semiparametric or parametric models for the MTE and the propensity score, depending on the extent of extrapolation required. 

To estimate optimal encouragement rules, I apply the social welfare criterion of encouragement rules, identified via the MTE function, to one popular class of statistical decision rules: \emph{Empirical Welfare Maximization} (EWM) rules \citep[see][Section 2.3]{hirano2020asymptotic}. The EWM approach directly chooses an optimal policy from a constrained class of feasible policies based on sample data. Constraints naturally arise in realistic settings when the policymaker wants to avoid complicated rules or satisfy legal, ethical, or political considerations. The EWM approach has proven to be practically implementable.\footnote{Implementable algorithms include mixed integer linear programming \citep{kitagawa2018should} and policy trees \citep{athey2021policy}. EWM rules with mandatory treatment assignment have been implemented in empirical applications using data from the National Job Training Partnership Act (JTPA) Study \citep{kitagawa2018should,kitagawa2021equality,mbakop2021model,sasaki2020welfare}, the Oregon Health Insurance Experiment (OHIE) \citep{sun2021empirical}, and the California Greater Avenues for Independence (GAIN) Program \citep{athey2021policy}, among others.} 
The second main result of this paper is to establish the convergence rate of the average welfare loss (regret) of the EWM encouragement rule relative to an oracle optimal rule, while allowing for a wide class of nonparametric and parametric estimators for the MTE and the propensity score. To keep the analysis tractable, I focus on settings in which the policymaker allocates individuals to two \emph{a priori} chosen manipulations of the instrument and constrains the class of feasible allocations.\footnote{In principle, it is possible to generalize the regret analysis to multi-action settings by incorporating different complexity measures of the policy class, such as the entropy integral used in \citet{zhou2023offline}, but a full analysis is beyond the scope of this paper.} 

I further consider two practically relevant extensions. First, when there are multiple instruments, I propose to work with a treatment selection model that allows for unobserved heterogeneity in the marginal rate of substitution across instruments. Second, when there is a budget constraint, I introduce the budget-constrained EWM encouragement rule and analyze its properties in terms of asymptotic optimality and asymptotic feasibility.

I apply the EWM encouragement rule to an empirical dataset from the third wave of the Indonesian Family Life Survey (IFLS). The goal is to provide advice on how upper secondary schooling can be encouraged to maximize average adult wages by manipulating the tuition fee. I find that the optimal policy without budget constraints provides tuition subsidy eligibility to individuals who face relatively high tuition fees and live relatively close to the nearest secondary school. I also provide a partial explanation of why this subpopulation is targeted.

\textbf{Related Literature}: This paper is related to four strands of literature, of which I provide a non-exhaustive overview below.

First, the research question is closely related to the literature on statistical treatment rules in econometrics following the seminal work of \citet{manski2004statistical}. 
See \citet{hirano2020asymptotic} for a recent review. Despite the breadth of the literature, only a few works look into observational data settings when the unconfoundedness assumption does not hold. \citet{kasy2016partial} and \citet{byambadalai2022} focus on cases of partial identification and welfare ranking of policies rather than optimal policy choices. \citet{athey2021policy} assume homogeneous treatment effects so that the \emph{conditional average treatment effect} (CATE) on compliers can be extrapolated to those on the entire population. \citet{sasaki2020welfare} identify the social welfare criterion via the MTE and demonstrate an application to the EWM framework. Nonetheless, these works implicitly assume complete enforcement of treatment rules, whereas I consider more realistic policy tools. In particular, the social welfare representation for treatment rules in \citet{sasaki2020welfare} can be viewed as a special case of that for binary encouragement rules in this paper when manipulations of the instrument are extremely strong. This echoes the discussion in their Section 3 of the presumption of full compliance under the treatment assignment, which they recognize will be rationalized in extreme circumstances.

The only exception I am aware of is \citet{chen2022personalized}, who use the MTE framework to study the personalized subsidy rule. Their work is complementary to mine in emphasizing different aspects of policy learning. They focus on the oracle optimal policy without restricting the policy class, whereas I analyze the estimated optimal policy within a restricted policy class. Notably, their closed-form characterization of optimal subsidy rules requires monotonicity of the MTE function with respect to the selection unobservable, which can be restrictive in practice. For example, in the context of the effects of family size on child outcomes, the quantity-quality model of fertility by \citet{becker1973interaction} is consistent with both positive and negative effects of family size depending on the level of complementarity in parental preferences between quantity and quality of children. Consequently, a monotone MTE function can mask important heterogeneity. Indeed, the MTE estimates in \citet{brinch2017beyond} show a U shape. In contrast, my framework allows for a flexible form of the MTE function. 

Second, in epidemiology and biostatistics, there has been increasing interest in \emph{individualized treatment rules}. \citet{cui2020semiparametric} and \citet{qiu2020optimal} allow for treatment endogeneity. They achieve point identification by leveraging the ``no common effect modifier'' assumption outlined in \citet{wang2018bounded}, which largely restricts the heterogeneity of compliance behavior. \citet{pu2021estimating} introduce the notion of ``IV-optimality'' to estimate the optimal treatment regime based on partial identification of the CATE. Unlike my framework, these works do not account for imperfect enforcement as a consequence of treatment endogeneity. One exception is \citet{qiu2020optimal}, who consider \emph{individualized encouragement rules} that manipulate a binary instrument. My framework nests theirs by allowing the instrument to have richer support.

Third, if one entirely discards information about the treatment and focuses on the relationship between the instrument and the outcome, as in the \emph{intention-to-treat} analysis, then the optimal manipulation of the instrument can be studied within a policy learning framework with general action spaces. Existing works have covered settings with multivalued actions \citep{zhou2023offline,fang2025model} and continuous actions \citep{kallus2018policy,ai2024data}. These works assume unconfoundedness and strong overlap, so that identification of the average counterfactual outcome is not an issue. In contrast, my approach explicitly incorporates information about the treatment and imposes additional structure, namely that the instrument affects the outcome only through the treatment. Under this structure, the average counterfactual outcome depends on the MTE and the propensity score in an interpretable way, which also enables rigorous extrapolation away from the instrument variation observed in the data.

Lastly, the representation of the social welfare criterion in this paper can be viewed as a variation of \emph{policy relevant treatment effects} (PRTE), adding to the class of policy parameters that can be written as weighted averages of the MTE. Hence, this paper complements the literature on PRTE, including \citet{carneiro2010evaluating,carneiro2011estimating}, \citet{mogstad2018using}, and \citet{sasaki2021estimation}, among many others.

\textbf{Organization}: The rest of the paper is organized as follows. Section \ref{encouragementcontinuous} sets up the model, introduces the encouragement rule, and derives a representation of the social welfare criterion via the MTE. It also discusses the identification of the social welfare criterion based on this representation and elaborates on binary encouragement rules.
Section \ref{applytoEWM} applies the social welfare representation to the EWM method and analyzes the regret properties. Section \ref{extensions} discusses extensions incorporating multiple instruments and budget constraints. Section \ref{empirical} presents an empirical application. Section \ref{conclusion} concludes. Proofs and additional results are collected in the appendix.

\section{Encouragement Rules with An Instrument}
\label{encouragementcontinuous}
\subsection{Setup}
\label{setup}

I consider the canonical program evaluation problem with a binary treatment $D\in\{0,1\}$ and a scalar, real-valued outcome $Y\in\mathcal{Y}\subset\mathbb{R}$. Outcome production is modeled through the potential outcomes framework \citep{rubin1974estimating}:
\begin{equation*}
    Y=Y(1)D+Y(0)(1-D),
\end{equation*}
where $(Y(0),Y(1))$ are the potential outcomes under no treatment and under treatment.\footnote{By adopting the potential outcomes model, I implicitly follow the conventional practice of imposing the \emph{Stable Unit Treatment Value Assumption} (SUTVA), namely that there are no spillover or general equilibrium effects.} Let $X\in\mathcal{X}\subset \mathbb{R}^{d_x}$ denote a vector of pretreatment covariates. For instance, in the analysis of returns to schooling, $D$ is an indicator for school enrollment, $Y$ is the log wage, and $X$ includes observable characteristics that affect wages (e.g., parental education, rural/urban residence). 

A (non-randomized) \textit{treatment rule} is defined as a mapping $\pi:\mathcal{X}\to\{0,1\}$. The policymaker’s objective function is the utilitarian (additive) welfare criterion defined by the average counterfactual outcome:
\begin{equation*}
    W(\pi)=E[Y(\pi(X))]=E[Y(1)\cdot\pi(X)+Y(0)\cdot(1-\pi(X))].
\end{equation*}
Define the \textit{conditional average treatment response} functions as $\mu_d(x)=E[Y(d)|X=x],d=0,1$. By the law of iterated expectations,
\begin{equation*}
    W(\pi)=E[\mu_1(X)\cdot\pi(X)+\mu_0(X)\cdot(1-\pi(X))].
\end{equation*}
Under the unconfoundedness assumption that $D$ is independent of $(Y(0),Y(1))$ conditional on $X$, $\mu_d(x)$ is identified by $E[Y|D=d,X=x]$ for $d=0,1$. However, the unconfoundedness assumption is violated if, for example, people self-select into schooling based on unmeasured benefits and costs driven by ability and motivation, both of which also affect wages. As a result, $\mu_d(x)\neq E[Y|D=d,X=x]$ for $d=0,1$ in general, and thus the social welfare criterion is not identified by the moments of observables. In this case, it is helpful to assume that there exists an instrument (i.e., an excluded variable) $Z\in\mathcal{Z}\subset \mathbb{R}$ that affects the treatment but not the outcome, e.g., the tuition fee. For each $z\in\mathcal{Z}$, denote the potential treatment status if the instrument were set to $z$ by $D(z)$. The observed treatment is given by $D=D(Z)$. I explicitly model the selection into treatment via an additively separable latent index model:
\begin{equation}
   D(z)=1\{\tilde{\nu}(X,z)-\tilde{U}\geq0\},\label{treatment}
\end{equation}
where $\tilde{\nu}$ is an unknown function, and $\tilde{U}$ represents unobservable factors that affect treatment choice. Let ``$\perp$'' denote (conditional) statistical independence. I adopt the following assumptions from the MTE literature \citep{heckman2005structural,mogstad2018using}:

\begin{enumerate}[label=\textbf{Assumption \arabic*},ref=\arabic*,itemindent=5\parindent,leftmargin=0pt]
\item \label{MTErestrictions}
(IV Restrictions and Continuous Distribution) 
\begin{enumerate}[label=(\roman*)]
    \item \label{MTEexogeneity} 
    $\tilde{U}\perp Z|X$.
    \item \label{MTEexclusion} 
    $E[Y(d)|X,Z,\tilde{U}]=E[Y(d)|X,\tilde{U}]$ and $E[|Y(d)|]<\infty$ for $d\in\{0,1\}$.
    \item \label{continuousU} 
    $\tilde{U}$ is continuously distributed conditional on $X$.
\end{enumerate}
\end{enumerate}
Assumptions \ref{MTErestrictions}(i) and (ii) impose exogeneity and an exclusion restriction on $Z$ but allow for arbitrary dependence between $(Y(0),Y(1))$ and $\tilde{U}$, even conditional on $X$. \citet{vytlacil2002independence} shows that, under Assumption \ref{MTErestrictions}(i), the existence of an additively separable selection equation as in (\ref{treatment}) is equivalent to the monotonicity assumption used for the \emph{local average treatment effects} (LATE) model of \citet{imbens1994identification}. The LATE monotonicity assumption restricts choice behavior in the sense that, conditional on $X$, an exogenous shift in $Z$ either weakly encourages or discourages every individual to choose $D=1$. Nonetheless, I maintain the selection equation (\ref{treatment}) because it allows me to express the average counterfactual outcome as a function of identifiable and interpretable objects. Under Assumptions \ref{MTErestrictions}(i) and (iii), one can reparameterize the model as
\begin{equation}
    D(z)=1\{\nu(X,z)-U\geq0\}\quad\text{with}\quad U|X,Z\sim \operatorname{Unif}[0,1],
    \label{selection}
\end{equation}
where $U\equiv F_{\tilde{U}|X}(\tilde{U}|X)$ and $\nu(x,z)\equiv F_{\tilde{U}|X}(\tilde{\nu}(x,z)|x)$. As a consequence,
\begin{equation*}
    p(x,z)\equiv\Pr(D=1|X=x,Z=z)=F_{U|X,Z}(\nu(x,z))=\nu(x,z),
\end{equation*}
where $p(x,z)$ is the propensity score.

Endogenous treatment selection also challenges the plausibility of fully mandating treatment assignment. I instead consider a different class of policies that manipulate the instrument, which I refer to as \emph{encouragement rules}. Formally, an encouragement rule is a mapping $\boldsymbol{\alpha}:\mathcal{X}\times\mathcal{Z}\to\mathbb{R}$ that manipulates the instrument for an individual with $(X,Z)=(x,z)$ from the initial value $z$ to a new level $\boldsymbol{\alpha}(x,z)$. For example, when $Z$ is the tuition fee, $\boldsymbol{\alpha}(x,z)=(z-\alpha(x))\cdot1\{z\geq\alpha(x)\}$ with $\alpha:\mathcal{X}\to\mathbb{R}_+$ describes a tuition subsidy rule that subsidizes an individual with $X=x$ up to $\alpha(x)$. The representation of the social welfare criterion in Section \ref{secrepresentation} covers the most general setting without restrictions on $\boldsymbol{\alpha}$. When I study the regret bounds for statistical decision rules in Section \ref{applytoEWM}, I take a stand on the complexity of feasible encouragement rules. In particular, I focus on a case in which the policymaker allocates individuals to two \emph{a priori} chosen manipulations of the instrument. The binary formulation of encouragement rules nests treatment rules \citep{kitagawa2018should,sasaki2020welfare} as a special case.

\subsection{Representation of the Social Welfare Criterion via the MTE}
\label{secrepresentation}

The outcome that would be observed under encouragement rule $\boldsymbol{\alpha}$ is
\begin{equation*}
    Y(D(\boldsymbol{\alpha}(X,Z)))=Y(1)\cdot D(\boldsymbol{\alpha}(X,Z))+Y(0)\cdot (1-D(\boldsymbol{\alpha}(X,Z))).
\end{equation*}
I define the social welfare criterion as the average counterfactual outcome: 
\begin{equation*}
    W(\boldsymbol{\alpha})\equiv E[Y(D(\boldsymbol{\alpha}(X,Z)))].
\end{equation*}
Theorem \ref{representation} shows that $W(\boldsymbol{\alpha})$ can be expressed as a function of the average treatment effect conditional on observable characteristics $X$ and the selection unobservable $U$, which is the definition of the MTE:
\begin{equation*}
    \operatorname{MTE}(x,u)\equiv E[Y(1)-Y(0)|X=x,U=u].
\end{equation*}
A proof is provided in \ref{proofs}. The concept of MTE was introduced by \citet{bjorklund1987estimation} and extended by \citet{heckman2005structural,heckman2007econometric}. In contrast to the intention-to-treat approach, the representation in Theorem \ref{representation} has a straightforward interpretation: among the individuals with $(X,Z)=(x,z)$, those for whom the value $u$ of $U$ is between $p(x,\boldsymbol{\alpha}(x,z))$ and $p(x,z)$ get either encouraged or discouraged to take up the treatment, and their contribution to the welfare contrast $W(\boldsymbol{\alpha})-E[Y]$ is $\operatorname{MTE}(x,u)$ if encouraged and $-\operatorname{MTE}(x,u)$ if discouraged.
    
\begin{theorem}\label{representation}
Under Assumption \ref{MTErestrictions}, the social welfare criterion for a given encouragement rule $\boldsymbol{\alpha}$ is given by
\begin{equation*}
    W(\boldsymbol{\alpha})=E[Y]+E\Big[\int_0^1 \operatorname{MTE}(X,u)\cdot(1\{p(X,\boldsymbol{\alpha}(X,Z))\geq u\}-1\{p(X,Z)\geq u\})\,du\Big].
\end{equation*}
\end{theorem}

A similar representation result appears in \citet[Proposition 1]{chen2022personalized}, where they exclude $Z$ from the targeting variables. This exclusion narrows their focus to policies that induce a degenerate distribution of $Z$ conditional on $X=x$. In contrast, my framework also accommodates policies that shift the conditional distribution of $Z$ through deterministic transformations, e.g., $\boldsymbol{\alpha}(x,z)=z+\alpha(x)$ with $\alpha:\mathcal{X}\to\mathbb{R}$.

It is worth noting that $W(\boldsymbol{\alpha})$ has a natural connection to the concept of policy relevant treatment effects (PRTE) \citep{heckman2005structural}. For a general class of policies that affect the propensity score, the PRTE is defined as the mean effect of going from a baseline policy to an alternative policy per net person shifted (assuming that $E[D|\text{alternative policy}]-E[D|\text{baseline policy}]\neq0$):
\begin{equation*}
    \frac{E[Y|\text{alternative policy}]-E[Y|\text{baseline policy}]}{E[D|\text{alternative policy}]-E[D|\text{baseline policy}]}.
\end{equation*}
Corollary \ref{PRTE} gives an alternative representation of $W(\boldsymbol{\alpha})$ in terms of a suitably defined PRTE.

\begin{corollary}\label{PRTE}
Under Assumption \ref{MTErestrictions}, when $E[p(X,\boldsymbol{\alpha}(X,Z))]-E[p(X,Z)]\neq0$, the social welfare criterion for a given encouragement rule $\boldsymbol{\alpha}$ is given by
\begin{equation*}
W(\boldsymbol{\alpha})=E[Y]+(E[p(X,\boldsymbol{\alpha}(X,Z))]-E[p(X,Z)])\cdot\operatorname{PRTE}(\boldsymbol{\alpha}),
\end{equation*}
where
\begin{equation*}
    \operatorname{PRTE}(\boldsymbol{\alpha})=E\Big[\int_0^1 \operatorname{MTE}(X,u)\cdot \omega(X,Z,u;\boldsymbol{\alpha})\,du\Big]
\end{equation*}
with the weight defined as
\begin{equation*}
    \omega(x,z,u;\boldsymbol{\alpha})\equiv\frac{1\{p(x,\boldsymbol{\alpha}(x,z))\geq u\}-1\{p(x,z)\geq u\}}{E[p(X,\boldsymbol{\alpha}(X,Z))]-E[p(X,Z)]}.
\end{equation*}
\end{corollary}

Corollary \ref{PRTE} unfolds the two forces driving the optimal policy based on $W(\boldsymbol{\alpha})$: the average change in treatment take-up, $E[p(X,\boldsymbol{\alpha}(X,Z))]-E[p(X,Z)]$, and the average treatment effect among those induced to switch treatment status, $\operatorname{PRTE}(\boldsymbol{\alpha})$, when going from the status quo to encouragement rule  $\boldsymbol{\alpha}$. Moreover, $\operatorname{PRTE}(\boldsymbol{\alpha})$ can be expressed as a weighted average of the MTE with weights determined by both observed and unobserved heterogeneity in treatment take-up.

\subsection{Identification of the Social Welfare Criterion}
\label{secidentification}

Theorem \ref{representation} implies that the point-identification of $W(\boldsymbol{\alpha})$ is guaranteed by the point-identification of the propensity score and the MTE over necessary domains. I formalize this insight in the following assumption.

\begin{enumerate}[label=\textbf{Assumption \arabic*},ref=\arabic*,itemindent=5\parindent,leftmargin=0pt]
\setcounter{enumi}{1}
\item \label{identifyW} (Point-Identification of $W(\boldsymbol{\alpha})$) 
\begin{enumerate}[label=(\roman*)]
    \item \label{identifyp} 
    $p(x,z)$ is point-identified over $\operatorname{Supp}(X,\boldsymbol{\alpha}(X,Z))$.
    \item \label{identifyMTE} 
    For every $x\in\mathcal{X}$, $\operatorname{MTE}(x,\cdot)$ is point-identified over $[\min\mathcal{P}_{\boldsymbol{\alpha}}(x),\max\mathcal{P}_{\boldsymbol{\alpha}}(x)]$, where $\mathcal{P}_{\boldsymbol{\alpha}}(x)$ denotes the support of $p(X,\boldsymbol{\alpha}(X,Z))$ conditional on $X=x$.
\end{enumerate}
\end{enumerate}

The method of local instrumental variables (LIV) \citep{heckman1999local,heckman2001local} gives
\begin{equation}
    \operatorname{MTE}(x,u)=\frac{\partial}{\partial u}E[Y|X=x,p(X,Z)=u],\label{LIV}
\end{equation}
provided that $u\mapsto E[Y|X=x,p(X,Z)=u]$ is continuously differentiable for almost every $x$. Therefore, Assumption \ref{identifyW}(ii) is satisfied via the point-identification of the derivative of $E[Y|X=x,p(X,Z)=u]$ with respect to $u$. As a result, $W(\boldsymbol{\alpha})$ is point-identified as
\begin{equation*}
    W(\boldsymbol{\alpha})=E[Y]+E[\mu_Y(X,p(X,\boldsymbol{\alpha}(X,Z)))-\mu_Y(X,p(X,Z))],
\end{equation*}
where $\mu_Y(x,u)\equiv E[Y|X=x,p(X,Z)=u]$.

I give three examples of sufficient conditions for Assumption \ref{identifyW}. Verification of Assumption \ref{identifyW} in these examples is relegated to \ref{verify}. Typically, additional structural restrictions are required to compensate for the relaxation of the support condition.

\begin{example}[Nonparametric Identification]\label{examplenonpara}
Assume that (E1-1) $\operatorname{Supp}(X,\boldsymbol{\alpha}(X,Z))\subset\operatorname{Supp}(X,\allowbreak Z)$; (E1-2) the conditional distribution of $p(X,Z)$ given $X$ is absolutely continuous with respect to the Lebesgue measure. Then, Assumption \ref{identifyW} is satisfied.
\end{example}

\begin{example}[Semiparametric Identification]\label{examplesemipara}
Assume that (E2-1) $\operatorname{Supp}(\boldsymbol{\alpha}(X,Z))\subset\mathcal{Z}$; (E2-2) the distribution of $p(X,Z)$ is absolutely continuous with respect to the Lebesgue measure; (E2-3) the propensity score is modeled as $p(x,z)=x^\top\gamma+\theta(z)$, where $\gamma$ is an unknown parameter and $\theta$ is an unknown function; (E2-4) the potential outcomes are modeled as $Y(0)=X^\top\beta_0+V_0$ and $Y(1)=X^\top\beta_1+V_1$, where $\beta_1$ and $\beta_0$ are unknown parameters, and $E[V_d|X,Z,U]=E[V_d|U]$ for $d=0,1$;\footnote{The partially linear form for potential outcomes in (E2-4) are commonly assumed in applied work estimating the MTE; see, e.g., \citet{carneiro2009estimating}, \citet{carneiro2010evaluating,carneiro2011estimating}. These works also invoke full independence $(U,V_0,V_1)\perp(Z,X)$, which is stronger than the conditional mean independence of $(V_0,V_1)$ from $(Z,X)$ in (E2-4).} (E2-5) $E[(X-E[X|Z])(X-E[X|Z])^\top]$ and $E[(X-E[X|p(X,Z)])(X-E[X|p(X,Z)])^\top]$ are positive definite. Then, Assumption \ref{identifyW} is satisfied. When (E2-1) is violated, Assumption \ref{identifyW} can still be satisfied by imposing a parametric model for the propensity score and a semiparametric partially linear model for the MTE, which is what I implement in the empirical application.
\end{example}

\begin{example}[Parametric Identification]\label{examplepara} 
Assume that (E3-1) the propensity score is modeled as $p(x,z)=b(x,z)^\top\gamma$, where $b$ is a known vector function and $\gamma$ is an unknown parameter;\footnote{Alternatively, one may adopt a logit or probit model to respect the $[0,1]$ boundary.} (E3-2) the conditional mean of $Y$ given $(X,p(X,Z))$ is modeled as $E[Y|X=x,p(X,Z)=u]=ux^\top\beta_1+(1-u)x^\top\beta_0+\sum_{j=2}^J \eta_ju^j$, where $\beta_1$, $\beta_0$, and $\eta_2,\dots,\eta_J$ are unknown parameters; (E3-3) there is no multicollinearity in $b(X,Z)$ elements nor in $(p(X,Z)X^\top,\allowbreak(1-p(X,Z))X^\top,p(X,Z)^2,\allowbreak\dots,p(X,Z)^J)$.\footnote{Polynomial MTE models are often used in empirical studies; see, e.g., \citet{brinch2017beyond}, \citet{cornelissen2018benefits}.} Then, Assumption \ref{identifyW} is satisfied.
\end{example}

\subsection{Binary Encouragement Rules}
\label{binaryencouragement}

When studying the regret bounds for statistical decision rules in Section \ref{applytoEWM}, I focus on settings in which the policymaker allocates individuals to two \emph{a priori} chosen manipulations of the instrument. Formally, given two functions $\alpha_0,\alpha_1:\mathcal{X}\times\mathcal{Z}\to\mathbb{R}$, a \textit{binary encouragement rule}, indexed by a mapping $\pi:\mathcal{X}\times\mathcal{Z}\to\{0,1\}$, manipulates the instrument for an individual with $(X,Z)=(x,z)$ to 
\begin{equation*}
\boldsymbol{\alpha}^\pi(x,z)=\pi(x,z)\cdot\alpha_1(x,z)+(1-\pi(x,z))\cdot\alpha_0(x,z).
\end{equation*}
Corollary \ref{representationbinary} specializes Theorem \ref{representation} to the binary setting.

\begin{corollary}\label{representationbinary}
Under Assumption \ref{MTErestrictions}, the social welfare criterion for a given binary encouragement rule $\pi$ is given by
\begin{eqnarray*}
    W(\boldsymbol{\alpha}^\pi)&=&E[Y(D(\alpha_0(X,Z)))]+E\Big[\pi(X,Z)\\
    &&\cdot\int_0^1 \operatorname{MTE}(X,u)\cdot(1\{p(X,\alpha_1(X,Z))\geq u\}-1\{p(X,\alpha_0(X,Z))\geq u\})\,du\Big].
\end{eqnarray*}
\end{corollary}

Finally, I demonstrate that the binary formulation of encouragement rules nests as a special case treatment rules that directly assign individuals to a certain treatment status. Suppose that $\alpha_0$ and $\alpha_1$ satisfy
\begin{equation*}
    p(X,\alpha_1(X,Z))=1 \text{ and }  p(X,\alpha_0(X,Z))=0 \text{ almost surely }
\end{equation*}
so that $\alpha_1$ (resp. $\alpha_0$) creates perfectly strong incentives (resp. disincentives) to be in the treated state ($D=1$) across heterogeneous covariate values. In this case, encouragement rules are effectively treatment rules: $D(\boldsymbol{\alpha}^\pi(X,Z))=\pi(X,Z)$.\footnote{In this case, $Z$ is redundant as a targeting variable.} Therefore, Corollary \ref{representationbinary} provides a representation of the social welfare criterion for treatment rules via the MTE function:
\begin{equation*}
    W(\boldsymbol{\alpha}^\pi)=E[Y(0)]+E\Big[\pi(X,Z)\int_0^1 \operatorname{MTE}(X,u)\,du\Big].
\end{equation*}
This representation coincides with Theorem 1 of \citet{sasaki2020welfare}. However, such powerful manipulations are hard to justify in practice. For example, consider a selection of the form $D=1\{Z\geq \tilde{U}\}$ with $\tilde{U}$ having full support on $\mathbb{R}$ and a manipulation of the form $\alpha_d(x,z)=z+a_d$ for $d=0,1$. Then, one needs to set $a_1=\infty$ and $a_0=-\infty$ to induce full compliance. Indeed, in Section 3 of \citet{sasaki2020welfare}, they recognize that the presumption of full compliance under the treatment assignment will be rationalized in extreme circumstances such as strong legal power or a large amount of resources held by the policymaker.

\section{Applications to EWM and Regret Properties}
\label{applytoEWM}

In this section, I restrict attention to binary encouragement rules described in Section \ref{binaryencouragement}. I apply the social welfare criterion, identified via the MTE function, to the EWM framework and investigate the theoretical properties of the resulting statistical decision rules.

Some extra notations are needed to facilitate the discussion. Suppose the policymaker observes a random sample $A_i=(Y_i,D_i,X_i,Z_i)$ of size $n$. Let $E_n$ denote the sample average operator, i.e., $E_n f=\frac{1}{n}\sum_{i=1}^n f(A_i)$ for any measurable function $f$. Let $a\vee b=\max\{a,b\}$. 

For notational simplicity, denote an encouragement rule and its social welfare by $\pi$ and $W(\pi)$ in place of $\boldsymbol{\alpha}^\pi$ and $W(\boldsymbol{\alpha}^\pi)$, respectively. Let $\Pi$ denote the class of encouragement rules the policymaker can choose from. In view of Corollary \ref{representationbinary} and (\ref{LIV}), $W(\pi)$ is point-identified under Assumption \ref{identifyW} as
\begin{equation*}
    W(\pi)=E[Y(D(\alpha_0(X,Z)))]+E[\pi(X,Z)\cdot\{\mu_Y(X,p(X,\alpha_1(X,Z)))-\mu_Y(X,p(X,\alpha_0(X,Z)))\}].
\end{equation*}
Define the welfare contrast relative to the baseline policy that allocates everyone to $\alpha_0$ as $\bar{W}(\pi)=W(\pi)-E[Y(D(\alpha_0(X,Z)))]$. The optimal encouragement rule is given by $\pi^*\in\argmax_{\pi\in\Pi}\bar{W}(\pi)$ if the distribution of $(X,Z)$ and the mappings $(x,z)\mapsto p(x,z)$ and $(x,u)\mapsto \mu_Y(x,u)$ are known. However, these quantities are unknown in practice. Given an estimator $\hat{p}(x,z)$ for $p(x,z)$ and an estimator $\hat{\mu}_Y(x,u)$ for $\mu_Y(x,u)$, I construct the empirical welfare criterion $\hat{W}_n(\pi)$ by plugging in these estimators:
\begin{equation*}
    \hat{W}_n(\pi)=
   E_n[\pi(X,Z)\cdot\{\hat{\mu}_Y(X,\hat{p}(X,\alpha_1(X,Z)))-\hat{\mu}_Y(X,\hat{p}(X,\alpha_0(X,Z)))\}].
\end{equation*}
Then, I define the \textit{feasible EWM encouragement rule} as 
\begin{equation}
\hat{\pi}_{\mathrm{FEWM}}\in\argmax_{\pi\in\Pi}\hat{W}_n(\pi).
\label{feasiblerule}
\end{equation}

In line with the literature on statistical treatment rules \citep{manski2004statistical,kitagawa2018should}, I evaluate the performance of an encouragement rule $\pi$ by its \textit{regret} defined as the welfare loss relative to the highest attainable welfare within class $\Pi$:
\begin{equation*}
R(\pi)=\max_{\pi'\in\Pi}W(\pi')-W(\pi).
\end{equation*}
To analyze the regret of $\hat{\pi}_{\mathrm{FEWM}}$, I impose the following assumptions.

\begin{enumerate}[label=\textbf{Assumption \arabic*},ref=\arabic*,itemindent=5\parindent,leftmargin=0pt]
\setcounter{enumi}{2}
\item \label{boundedvc} 
(Boundedness and Vapnik-Chervonenkis (VC)-Class) 
\begin{enumerate}[label=(\roman*)]
    \item \label{bounded} 
    There exists $\bar{M}<\infty$ such that $\sup_{(u,x)\in[0,1]\times\mathcal{X}}|\operatorname{MTE}(x,u)|\leq\bar{M}$.
    \item \label{vc} 
    $\Pi$ has a finite VC-dimension.
\end{enumerate}
\end{enumerate}
Assumption \ref{boundedvc}(i) requires the MTE to be uniformly bounded in $u$ and $x$. Assumption \ref{boundedvc}(ii) controls the complexity of the class $\Pi$ of candidate encouragement rules in terms of VC-dimension. Interested readers can refer to \citet{van1996weak} for the definition and textbook treatment of VC-dimension. I now give two examples of $\Pi$ that satisfy Assumption \ref{boundedvc}(ii).

\begin{example}[Linear Eligibility Score (LES)]\label{exampleles} 
Let $v\in\mathbb{R}^{d_v}$ be a subvector of $(x,z)$. Consider the class of binary decision rules based on linear eligibility scores:
\begin{equation*}
    \Pi_{\mathrm{LES}}=\{\pi:\pi(x,z)=1\{\lambda_0+\lambda^\top v\geq0\},(\lambda_0,\lambda^\top)\in\mathbb{R}^{d_v+1}\}.
\end{equation*}
For example, individuals are assigned scholarships if a linear function of their tuition fee and distance to school exceeds some threshold. The EWM method searches over all possible linear coefficients. The VC-dimension of $\Pi_{\mathrm{LES}}$ is $d_v+1$.
\end{example}

\begin{example}[Threshold Allocations (TA)]\label{exampleta} 
Consider the class of binary decision rules based on threshold allocations:
\begin{equation*}
    \Pi_{\mathrm{TA}}=\{\pi:\pi(x,z)=1\{\sigma_k v_k\leq\bar{v}_k \text{ for } k\in\{1,\dots,d_v\}\},\bar{v}\in\mathbb{R}^{d_v},\sigma\in\{-1,1\}^{d_v}\}.
\end{equation*}
For example, individuals are assigned scholarships if their tuition fee and distance to school are above or below some thresholds. The EWM method searches over all possible thresholds and directions. The VC-dimension of $\Pi_{\mathrm{TA}}$ is $d_v$.
\end{example}

I also propose the following assumption about the unknown components that show up in the social welfare criterion: the propensity score $p(x,z)$ and the observed conditional average outcome $\mu_Y(x,u)$.

\begin{enumerate}[label=\textbf{Assumption \arabic*},ref=\arabic*,itemindent=5\parindent,leftmargin=0pt]
\setcounter{enumi}{3}
\item \label{estimatorpMTE} 
(Estimation of the Propensity Score and the Observed Conditional Average Outcome) 
\begin{enumerate}[label=(\roman*)]
    \item \label{supnormp}
    There exists a sequence $\psi_n\to\infty$ such that for each $d\in\{0,1\}$,
    \begin{equation*}
       E[|\hat{p}(X,\alpha_d(X,Z))-p(X,\alpha_d(X,Z))|]=O(\psi_n^{-1}).
    \end{equation*}
    \item \label{supnormMTE}
    There exists a sequence $\phi_n\to\infty$ such that
    \begin{equation*}
       E\Big[\sup_{u\in[0,1]}|\hat{\mu}_Y(X,u)-\mu_Y(X,u)|\Big]=O(\phi_n^{-1}).
    \end{equation*}
\end{enumerate}
\end{enumerate}
Assumption \ref{estimatorpMTE}(i) concerns the convergence rate in expectation of the estimation error for $p(x,z)$. When $p(x,z)$ is estimated nonparametrically as in Example \ref{examplenonpara}, a sufficient condition for Assumption \ref{estimatorpMTE}(i) is
$E[\sup_{(x,z)\in\operatorname{Supp}(X,Z)}|\hat{p}(x,z)-p(x,z)|]=O(\psi_n^{-1})$. In \ref{supnorm}, I derive the sup-norm convergence rate in expectation for local polynomial estimators and series estimators built on exponential tail bounds. The rate can be faster if a semiparametric or parametric estimator is used under additional assumptions as in Example \ref{examplesemipara} or \ref{examplepara}. 
Assumption \ref{estimatorpMTE}(ii) concerns the convergence rate in expectation of the estimation error for $\mu_Y(x,u)$. Since $U$ is not observed, I take the supremum over the unit interval. Usually, $\mu_Y(x,u)$ is estimated using the estimated propensity score as a generated regressor. When the regression model is parametric, I provide sufficient conditions for Assumption \ref{estimatorpMTE}(ii) to hold with $\phi_n=\psi_n$ in \ref{parametricMTE}. When the regression model is nonparametric, the sup-norm convergence rate in probability is established in \citet[Corollary 1]{mammen2012nonparametric}. However, the sup-norm convergence rate in expectation remains unknown. I leave it for future work.

\begin{theorem}\label{upperFEWM}
Suppose that Assumptions \ref{MTErestrictions}-\ref{estimatorpMTE} hold. Then,
\begin{equation*}
    E[R(\hat{\pi}_{\mathrm{FEWM}})]=O(\psi_n^{-1}\vee\phi_n^{-1}\vee n^{-1/2}).
\end{equation*}
\end{theorem}

Theorem \ref{upperFEWM} derives a convergence rate upper bound for the average regret of $\hat{\pi}_{\mathrm{FEWM}}$. A proof is provided in \ref{proofs}. In general, the convergence rate upper bound is determined by $\psi_n^{-1}\vee\phi_n^{-1}$. 

\begin{remark}
There are two special cases where the $n^{-1/2}$ rate can be achieved.\footnote{While the convergence rate lower bound in the current context is not known, it is natural to conjecture that it is $O(n^{-1/2})$. I leave the formal analysis for future work.} One is to assume parametric forms for $p(x,z)$ and $\mu_Y(x,u)$ so that $\phi_n=\psi_n=n^{1/2}$. The other is to pursue a doubly robust approach in the spirit of \citet{athey2021policy}. The idea is to use an alternative social welfare criterion based on a doubly robust score, which is Neyman-orthogonal with respect to $p(x,z)$ and $\mu_Y(x,u)$. The details are given in \ref{doublyrobust}. An extra cost to pay for implementing the doubly robust score is the estimation of the joint density of $(X,Z)$, which can be challenging if the dimension of $X$ is large.
\end{remark}

\section{Extensions}
\label{extensions}

I consider two empirically relevant extensions to the baseline setup in Section \ref{setup}. In Section \ref{multiplecontinuousiv}, I allow for the presence of other instruments in addition to the one that can be manipulated. In Section \ref{budgetcontinuous}, I incorporate budget constraints. As a further extension, I consider encouragement rules with a binary instrument in \ref{encouragementbinary}.

\subsection{Multiple Instruments}
\label{multiplecontinuousiv}

In practice, the policymaker can observe multiple instruments, but only one of them can be used as the tool for policy intervention. For example, tuition subsidies and proximity to upper secondary schools are two instruments for enrollment in upper secondary school, but only the former can serve as an encouragement. More generally, I allow $Z$ to be $L$-dimensional. Let $Z_1\in\mathcal{Z}_1\subset\mathbb{R}$ be the instrument that can be intervened upon, and let $Z_{-1}$ collect all other $(L-1)$ components. An encouragement rule is a mapping $\boldsymbol{\alpha}:\mathcal{X}\times\mathcal{Z}\to\mathbb{R}$ that determines the manipulated level of $Z_1$ while leaving $Z_{-1}$ unchanged.

For $z_1\in\mathcal{Z}_1$, I construct a selection equation for the potential treatment status if $Z_1$ were set to $z_1$ while $Z_{-1}$ remained at its observed realization as
\begin{equation}
    D(z_1,Z_{-1})=1\{p(X,z_1,Z_{-1})\geq U_1\}\quad\text{with}\quad U_1|X,Z\sim \operatorname{Unif}[0,1],
    \label{marginalselection}
\end{equation}
where $U_1$ can be interpreted as a latent proneness to take the treatment, which is measured against the incentive (or disincentive) created by the manipulated instrument.\footnote{Since only one instrument is manipulated, I focus on the ``marginal'' selection behavior induced by this instrument conditional on the other instruments. If one is interested in policies that simultaneously manipulate multiple instruments, then a treatment selection model with multidimensional unobserved heterogeneity may be needed; see, e.g., \citet{ura2024policy}.} By (\ref{marginalselection}), I only impose restrictions along one margin of selection and thus are agnostic about unobserved heterogeneity in the marginal rate of substitution across instruments.\footnote{\citet{mogstad2021causal} use a random utility model to demonstrate that in the presence of multiple instruments, (\ref{selection}) implies homogeneity in the marginal rate of substitution. In contrast, (\ref{marginalselection}) does not impose such implicit homogeneity.} \citet{chen2022personalized} adhere to (\ref{selection}) when they deal with multiple instruments, thereby presenting the same social welfare representation as in the single-instrument case (i.e., Theorem \ref{representation}).

The outcome that would be observed under encouragement rule $\boldsymbol{\alpha}$ is
\begin{equation*}
    Y(D(\boldsymbol{\alpha}(X,Z),Z_{-1}))=Y(1)\cdot D(\boldsymbol{\alpha}(X,Z),Z_{-1})+Y(0)\cdot(1-D(\boldsymbol{\alpha}(X,Z),Z_{-1})).
\end{equation*}
Define the social welfare criterion as $W(\boldsymbol{\alpha})=E[Y(D(\boldsymbol{\alpha}(X,Z),Z_{-1}))]$. Since (\ref{marginalselection}) only imposes that $U_1$ is independent of $Z_1$ given $(X,Z_{-1})$, I accordingly replace Assumption \ref{MTErestrictions}(ii) with an exclusion restriction that only requires potential outcomes to be mean independent of $Z_1$ given $(X,Z_{-1})$.

\begin{enumerate}[label=\textbf{Assumption \arabic*},ref=\arabic*,itemindent=5\parindent,leftmargin=0pt]
\setcounter{enumi}{4}
\item \label{multipleMTErestrictions} 
(Instrument-Specific Exclusion Restriction) 
$E[Y(d)|X,Z,U_1]\allowbreak=E[Y(d)|X,\allowbreak Z_{-1},U_1]$ and $E[|Y(d)|]<\infty$ for $d\in\{0,1\}$.
\end{enumerate}
It turns out that $W(\boldsymbol{\alpha})$ can be expressed as a function of the instrument-specific MTE defined as
\begin{equation*}
    \operatorname{MTE}_1(x,z_{-1},u_1)\equiv E[Y(1)-Y(0)|X=x,Z_{-1}=z_{-1},U_1=u_1].
\end{equation*}
The expression is given in Corollary \ref{Wmultiple}. The analysis in Section \ref{applytoEWM} then applies. Heuristically, $\operatorname{MTE}_1$ is equivalent to the MTE function using the manipulated instrument and conditioning on the other instruments as covariates.

\begin{corollary}\label{Wmultiple}
    Under (\ref{marginalselection}) and Assumption \ref{multipleMTErestrictions}, the social welfare criterion for a given encouragement rule $\boldsymbol{\alpha}$ is given by
    \begin{eqnarray*}
    W(\boldsymbol{\alpha})&=&E[Y]+E\Big[\int_0^1\operatorname{MTE}_1(X,Z_{-1},u_1)\\
    &&\cdot(1\{p(X,\boldsymbol{\alpha}(X,Z),Z_{-1})\geq u_1\}-1\{p(X,Z)\geq u_1\})\,du_1\Big].
\end{eqnarray*}
\end{corollary}
  
\subsection{Budget Constraints}
\label{budgetcontinuous}

Manipulating the instrument can be costly, especially when the instrument is a monetary variable such as price. In practice, the policymaker often faces budget constraints and wants to prioritize encouragement for the individuals who will benefit the most. Incorporating budget constraints is of particular interest when the treatment effect is intrinsically positive. For example, \citet{dupas2014short} documents an experiment in Kenya that randomly assigned subsidized prices for a new health product. The treatment and outcome were indicators for the product's purchase and usage, respectively. The product was not available outside the experiment, so the potential outcome if not treated is identically equal to zero. Hence, the first-best decision rule was to assign the treatment, or an encouragement that induced one-way flows into treatment, to everyone. However, to preserve financial resources, in this scenario, the policymaker may wish to exclude individuals who are not likely to increase product usage, for example, because of low disease risks in their neighborhood.

Let $C:\mathcal{X}\times\mathcal{Z}\to\mathbb{R}_+$ be a user-chosen cost function that potentially depends on $\boldsymbol{\alpha}$. For example, $C(x,z)=|\boldsymbol{\alpha}(x,z)-z|$ is a direct measure of manipulation costs.\footnote{Depending on the context, additional costs can be embedded in the experimental design. For example, in the experiment documented in \citet{thornton2008demand}, besides the monetary incentives for learning HIV results (the instrument), there were considerably high costs for testing, counseling/giving results, and selling condoms (see Table 12).} For encouragement rule $\boldsymbol{\alpha}$, I define its budget by aggregating the costs for individuals who actually take up the treatment: $B(\boldsymbol{\alpha})=E[C(X,Z)\cdot D(\boldsymbol{\alpha}(X,Z))]$. I consider settings in which the policymaker faces a harsh budget constraint such that the cost of implementing any encouragement rule cannot exceed $\kappa$.

\begin{remark}
    The policymaker may only want to account for cost without imposing a fixed budget, which is the thought experiment considered in \citet{kitagawa2018should} and \citet{chen2022personalized}. In this case, one can redefine the social welfare criterion as $W(\boldsymbol{\alpha})-B(\boldsymbol{\alpha})$ to apply the analysis in Section \ref{applytoEWM}.
\end{remark}

As in Section \ref{applytoEWM}, I specialize to binary encouragement rules when discussing the performance of statistical decision rules and denote the budget by $B(\pi)$ in place of $B(\boldsymbol{\alpha}^\pi)$. Given a class $\Pi$ of feasible encouragement rules,\footnote{I implicitly assume that there exists $\pi\in\Pi$ such that $B(\pi)\leq \kappa$.} the policymaker now solves a constrained optimization problem:
\begin{equation*}
    \max_{\pi\in\Pi} W(\pi) \text{ s.t. } B(\pi)\leq\kappa.
\end{equation*}
Let $\pi^*_{\mathrm{B}}$ denote the oracle solution. I follow \citet{sun2021empirical} to introduce two desirable properties for statistical decision rules in the current setting: \textit{asymptotic optimality} and \textit{asymptotic feasibility}. Intuitively, with a large enough sample size, asymptotic optimality imposes that a statistical decision rule $\hat{\pi}$ is unlikely to achieve strictly lower welfare than $\pi^*_\mathrm{B}$, and asymptotic feasibility imposes that $\hat{\pi}$ is unlikely to strictly violate the budget constraint.

\begin{definition}\label{optfsb}
A statistical decision rule $\hat{\pi}$ is \textit{asymptotically optimal} if, for any $\epsilon>0$,
\begin{equation*}
    \limsup_{n\to\infty}\Pr(W(\hat{\pi})-W(\pi^*_{\mathrm{B}})<-\epsilon)=0.
\end{equation*}
A statistical decision rule $\hat{\pi}$ is \textit{asymptotically feasible} if, for any $\epsilon>0$,
\begin{equation*}
    \limsup_{n\to\infty}\Pr(B(\hat{\pi})-\kappa>\epsilon)=0.
\end{equation*}
\end{definition}

\begin{remark}
    Asymptotic optimality and asymptotic feasibility are defined asymmetrically in \citet{sun2021empirical}. On one hand, asymptotic optimality only requires the population welfare of a statistical decision rule to concentrate around the optimal value from below. On the other hand, asymptotic feasibility requires the statistical decision rule to satisfy the population budget constraint without any slackness and thus is extremely sensitive to sampling uncertainty. In consequence, \citet{sun2021empirical} proves the negative result that no statistical decision rule can uniformly satisfy both properties over a sufficiently rich class of data generating processes. In contrast, after revising the definition of asymptotic feasibility to be symmetric with that of asymptotic optimality, I show that it is possible to construct a statistical decision rule that simultaneously achieves both properties. 
\end{remark}

Note that by (\ref{selection}),
\begin{equation*}
    B(\pi)=E[C(X,Z)\cdot\{\pi(X,Z)\cdot p(X,\alpha_1(X,Z))+(1-\pi(X,Z))\cdot p(X,\alpha_0(X,Z))\}].
\end{equation*}
Define the \textit{budget-constrained EWM encouragement rule} defined as a solution to the sample version of the population constrained optimization problem:
\begin{equation}
    \hat{\pi}_{\mathrm{BEWM}}\in\argmax_{\pi\in\Pi}\hat{W}_n(\pi) \text{ s.t. } \hat{B}_n(\pi)\leq \kappa,
    \label{constrainedsa}
\end{equation}
where
\begin{equation*}
    \hat{B}_n(\pi)=E_n[C(X,Z)\cdot\{\pi(X,Z)\cdot \hat{p}(X,\alpha_1(X,Z))+(1-\pi(X,Z))\cdot \hat{p}(X,\alpha_0(X,Z))\}].
\end{equation*}
I set $\hat{\pi}_{\mathrm{BEWM}}=\emptyset$ if no $\pi\in\Pi$ satisfies $\hat{B}_n(\pi)\leq \kappa$. Theorem \ref{sampleanaopt} asserts that $\hat{\pi}_{\mathrm{BEWM}}$ satisfies both properties in Definition \ref{optfsb}. A proof is provided in \ref{proofs}.

\begin{theorem}\label{sampleanaopt}
Suppose that Assumptions \ref{MTErestrictions}--\ref{estimatorpMTE} hold, and that $C(x,z)$ is uniformly bounded in $x$ and $z$. Then, $\hat{\pi}_{\mathrm{BEWM}}$ is asymptotically optimal and asymptotically feasible.
\end{theorem}

\section{Empirical Application}
\label{empirical}

In this section, I apply the feasible EWM encouragement rule and the budget-constrained EWM encouragement rule to provide guidance on how to encourage upper secondary schooling, using data from the third wave of the Indonesian Family Life Survey (IFLS) fielded from June through November 2000. \citet{carneiro2017average} used this dataset to study the returns to upper secondary schooling in Indonesia. I follow \citet{carneiro2017average} in restricting my sample to males aged 25--60 who are employed and who have non-missing reported wage and schooling information. This subsample consists of 2,104 individuals.\footnote{The subsample used in \citet{carneiro2017average} does not contain the tuition fee variable, which plays a central role in my framework as the manipulatable instrument. Hence, I followed their descriptions to construct my subsample from raw data downloaded from the RAND Corporation website.}

I specify the relevant variables in my framework as follows. The outcome $Y$ is the log of hourly wages (in rupiah) constructed from self-reported monthly wages and hours worked per week. The treatment $D$ is an indicator of attendance of upper secondary school or higher, corresponding to 10 or more years of completed education. The first instrument $Z_1$ is the lowest fee per continuing student, in thousands of rupiah, among secondary schools in the community of current residence.\footnote{The term ``community'' refers to the lowest-level administrative division in Indonesia. A community can either be a \textit{desa} (village) or a \textit{kelurahan} (urban community).} The second instrument $Z_2$ is the distance, in kilometers, from the office of the community head of current residence to the nearest secondary school, which I define as the secondary school closest to the office of the community head.\footnote{The validity of an instrument constructed in this way can be controversial. Each individual's tuition and distance to school are based on their current residence rather than their residence at the time of the secondary schooling decision. Educated individuals may move to more urban areas with more schools and higher tuition fees. Nonetheless, I note that the instrumental variable independence assumption for unrestricted instruments has testable implications, which are the generalized instrumental inequalities proposed by \citet{kedagni2020generalized}. Using their tests, I do not find evidence against the independence assumption between potential earnings and tuition fees, or between potential earnings and distance to school.
} 
I treat $Z_1$ as manipulatable and $Z_2$ as not manipulatable. Collect $Z=(Z_1,Z_2)^\top$. The covariates $X$ include age, age squared, an indicator of rural residence, distance from the office of the community head of residence to the nearest health post, and indicators for religion, parental education, and the province of residence. Table \ref{samplestat} in \ref{addtab} presents sample averages of these variables.

I focus on binary encouragement rules that manipulate $Z_1$ according to $\boldsymbol{\alpha}^\pi(x,z)=\pi(x,z)\cdot \alpha_1(x,z)+(1-\pi(x,z))\cdot \alpha_0(x,z)$. For the binary decision $\pi$, I consider the class of linear rules based on $(z_1,z_2)$:
\begin{equation*}
    \Pi_{\mathrm{LES}}=
    \big\{\pi:\pi(x,z)=1\{\lambda_{0}+\lambda_{1}\cdot z_1+\lambda_{2}\cdot z_2>0\},\lambda_0,\lambda_1,\lambda_2\in\mathbb{R}\}.
\end{equation*}
I specify the manipulation function as $\alpha_1(x,z)=(z_1-a)\cdot 1\{z_1\geq a\}$ and $\alpha_0(x,z)=z_1$. Here, $\alpha_1$ describes a tuition subsidy of up to $a$ and $\alpha_0$ describes the status quo. The policymaker \emph{a priori} chooses from $a\in\{2.5,22.25\}$, which correspond to the sample median and maximum of $Z_1$, respectively. I specify the cost function as $C(x,z)=|\boldsymbol{\alpha}^\pi(x,z)-z_1|$ and the budget constraint as $\kappa=0.28$, which is about one-tenth of the average hourly wage.

The fact that $Z_1$ has discrete support (with 56 distinct values) violates the support condition for nonparametric or semiparametric identification of the propensity score.\footnote{When $a=2.5$, only 20 out of the 35 support points of $\alpha_1(X,Z)$ lie in the support of $Z_1$. When $a=22.25$, $\alpha_1(X,Z)$ is identically equal to 0, which falls outside the support of $Z_1$.} Therefore, I estimate the propensity score from a logit regression of $D$ on $X$, $Z$, $Z_1\cdot Z_2$, and interactions between $Z$ and $X$.\footnote{This specification of propensity score is an adaptation of that considered by \citet{carneiro2017average} and \citet{sasaki2021estimation}, who use a single instrument $Z_2$.} Although all elements of $X$ are discrete, they together provide sufficient variation in the propensity score for the semiparametric estimation of the MTE.\footnote{The estimated propensity score takes 1,782 distinct values that almost cover the full unit interval.} I specify the conditional mean of $Y$ given $X$, $Z_2$, and $p(X,Z)$ as
\begin{equation*}
    E[Y|X=x,Z_2=z_2,p(X,Z)=u]=u(x^\top,z_2)\beta_1+(1-u)(x^\top,z_2)\beta_0+G(u),
\end{equation*}
where $G(\cdot)$ is an unknown function. By Corollary \ref{Wmultiple}, the social welfare criterion of encouragement rule $\pi$ is identified as $W(\pi)=E[Y]+E[\pi(X,Z)\cdot((p(X,\alpha_1(X,Z),Z_2)-p(X,Z))(X^\top,Z_2)(\beta_1-\beta_0)+G(p(X,\alpha_1(X,Z),Z_2))-G(p(X,Z)))]$. I use the double residual regression procedure of \citet{robinson1988root} to estimate $(\beta_1,\beta_0)$. Given the estimators $\hat{p}(x,z)$ and $(\hat{\beta}_1,\hat{\beta}_0)$, I estimate $G(\cdot)$ using a nonparametric regression of the residual $Y-\hat{p}(X,Z)(X^\top,Z_2)\hat{\beta}_1-(1-\hat{p}(X,Z))(X^\top,Z_2)\hat{\beta}_0$ on $\hat{p}(X,Z)$. I use the locally linear regression throughout with a Gaussian kernel and a bandwidth of 0.06, which is determined by leave-one-out cross-validation.

I compute the feasible EWM encouragement rule $\hat{\pi}_{\mathrm{FEWM}}$ in (\ref{feasiblerule}) and the budget-constrained EWM encouragement rule $\hat{\pi}_{\mathrm{BEWM}}$ in (\ref{constrainedsa}) using the CPLEX mixed integer optimizer. Table \ref{welfaregain} presents point estimates of some key quantities of alternative encouragement rules. The first column reports the welfare gain, $W(\pi)-E[Y]$. The second column reports the share of eligible population for the tuition subsidy, $E[\pi(X,Z)]$. Based on the decomposition result in Corollary \ref{PRTE}, the third and fourth columns report the average change in treatment take-up, $E[p(X,\boldsymbol{\alpha}^\pi(X,Z),Z_2)]-E[p(X,Z)]$, and the PRTE, respectively. The former measures the proportion of individuals induced to enroll in or drop out of upper secondary school, and the latter measures the average change in the log of hourly wages among these individuals. 

\begin{table}[ht]
    \centering
    \caption{Comparison of Alternative Encouragement Rules}
    \begin{tabular}{l c c c c}
        \hline
        & & Share of Eligible & Avg. Change in &\\
       Policy & Welfare Gain & Population & Treatment Take-Up & PRTE\\
         \hline
\multicolumn{5}{l}{Panel A: $a=2.5$ (tuition subsidy up to the median tuition fee)}\\
$\hat{\pi}_{\mathrm{FEWM}}$ & 0.0146 & 0.388 & 0.0173 & 0.843\\
$\hat{\pi}_{\mathrm{BEWM}}$ ($\kappa=0.28$) & 0.0102 & 0.280 & 0.0137 & 0.743\\
$\pi(x,z)=1~\forall x,z$ & 0.0005 & 1 & 0.0022 & 0.230\\[6pt]
\multicolumn{5}{l}{Panel B: $a=22.25$ (full tuition waiver)}\\
$\hat{\pi}_{\mathrm{FEWM}}$ & 0.0218 & 0.386 & 0.0317 & 0.688\\
$\hat{\pi}_{\mathrm{BEWM}}$ ($\kappa=0.28$) & 0.0090 & 0.286 & 0.0141 & 0.639\\
$\pi(x,z)=1~\forall x,z$ & 0.0044 & 1 & 0.0140 & 0.315\\
\hline
    \end{tabular}
    \label{welfaregain}
    \begin{minipage}{\textwidth}
    \vspace{0.5em}
    {\footnotesize Notes: The first column reports the welfare gain, $W(\pi)-E[Y]$. The second column reports the share of eligible population for the tuition subsidy, $E[\pi(X,Z)]$. The third and fourth columns report the average change in treatment take-up, $E[p(X,\boldsymbol{\alpha}^\pi(X,Z),Z_2)]-E[p(X,Z)]$, and the PRTE, respectively, based on the decomposition result in Corollary \ref{PRTE}.\par}
    \end{minipage}
\end{table}

As can be seen from Table \ref{welfaregain}, the seemingly favorable tuition subsidy has little effect on overall upper secondary school attendance when applied to everyone, resulting in a welfare gain of only a small magnitude. In contrast, the feasible EWM encouragement rule and the budget-constrained EWM encouragement rule achieve higher welfare gains by targeting a subpopulation with both a greater increase in treatment take-up and higher PRTE.

I plot the feasible EWM encouragement rule and the budget-constrained EWM encouragement rule in Panels A and B of Figure \ref{ewmIFLS}, respectively. The shaded areas indicate the subpopulations to whom the tuition subsidy should be assigned. For both subsidy levels, the feasible EWM encouragement rule gives eligibility to individuals facing relatively high tuition fees and living relatively close to the nearest secondary school. The subpopulations targeted by the budget-constrained EWM encouragement rule shrink to the left. When the subsidy level $a$ is increased from 2.5 to 22.25, the budget-constrained EWM encouragement rule tends to prioritize individuals facing relatively low tuition fees.

\begin{figure}[ht]
    \centering
    \caption{Targeted Subpopulation under Alternative Encouragement Rules}
    \includegraphics{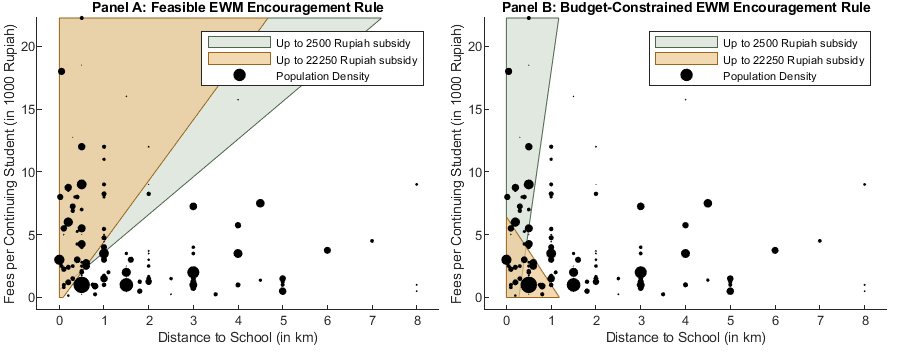}
    \label{ewmIFLS}
\end{figure}

Utilizing a decomposition of the social welfare criterion analogous to Corollary \ref{PRTE}, I offer a partial explanation of why the subpopulations indicated by the shaded areas in Figure \ref{ewmIFLS} are targeted. One can write
\begin{equation*}
   W(\pi)=E[Y]+E[\pi(X,Z)\cdot \operatorname{PRTE}(X,Z)\cdot(p(X,\alpha_1(X,Z),Z_2)-p(X,Z))],
\end{equation*}
where
\begin{equation*}
    \operatorname{PRTE}(x,z)=\frac{\int_{p(x,z)}^{p(x,\alpha_1(x,z),z_2)}\operatorname{MTE}_1(x,z_2,u_1)\,du_1}{p(x,\alpha_1(x,z),z_2)-p(x,z)}
\end{equation*}
measures the average treatment effect among individuals with $(X,Z)=(x,z)$ who are induced to switch treatment status when going from the status quo to the tuition subsidy $\alpha_1$. Let $\operatorname{Med}(X)$ denote the sample median of $X$. I focus on the case of $a=22.25$, i.e., a full tuition waiver. Figure \ref{interpretability} is based on point estimates of $p(x,z)$ and $\operatorname{PRTE}(x,z)$. Panel A displays the level sets of $(z_1,z_2)\mapsto p(\operatorname{Med}(X),\alpha_1(x,z),z_2)-p(\operatorname{Med}(X),z)$, namely the changes in treatment take-up for individuals with different values of $(Z_1,Z_2)$ and the median value of $X$. Panel B displays the level sets of $(z_1,z_2)\mapsto\operatorname{PRTE}(\operatorname{Med}(X),z)$.

\begin{figure}[ht]
    \centering
    \caption{Impact of Going from the Status Quo to a Full Tuition Waiver}
    \includegraphics{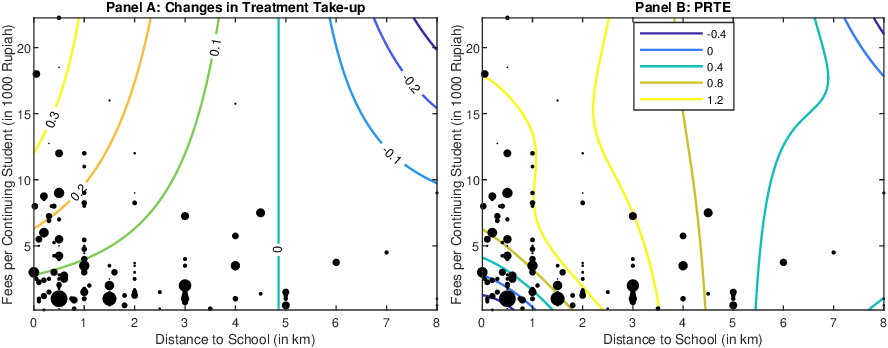}
    \label{interpretability}
    \begin{minipage}{\textwidth}
    \vspace{0.5em}
    {\footnotesize Notes: Panel A displays the level sets of $(z_1,z_2)\mapsto p(\operatorname{Med}(X),\alpha_1(x,z),z_2)-p(\operatorname{Med}(X),z)$. Panel B displays the level sets of $(z_1,z_2)\mapsto\operatorname{PRTE}(\operatorname{Med}(X),z)$. The size of the black dots indicates the number of individuals with different values of $(Z_1,Z_2)$.\par}
    \end{minipage}
\end{figure}

From Panel A of Figure \ref{interpretability}, it can be seen that only individuals with low $Z_2$ are induced into treatment. From Panel B of Figure \ref{interpretability}, it can be seen that individuals induced into treatment have positive PRTE except for those with low values of $(Z_1,Z_2)$. Put together, absent budget constraints, individuals in the upper-left corner are prioritized for the full tuition waiver. Meanwhile, contour lines where $Z_1$ is high are steep in both panels, implying that individuals with higher $Z_1$ incur greater manipulation costs for the same amount of welfare gains. Consequently, the optimal policy under budget constraints, which trades off welfare gains against costs, gives up individuals with high $Z_1$.

\section{Conclusion}
\label{conclusion}

In this paper, I propose a policy learning framework that allows for endogenous treatment selection by leveraging an instrumental variable. To deal with imperfect compliance when designing policies, I consider encouragement rules instead of treatment rules. To deal with failure of unconfoundedness when identifying the social welfare criterion, I incorporate the MTE function. Focusing on binary encouragement rules, I apply the representation of the social welfare criterion via the MTE to the EWM method and derive convergence rates of regret. I also consider extensions allowing for multiple instruments and budget constraints. I illustrate the EWM encouragement rule using data from the Indonesian Family Life Survey. 

To be clear, the analysis in this paper critically relies on the point-identification of the social welfare criterion. The necessary support condition or parametric assumptions could be restrictive. An interesting avenue for future research is to incorporate approaches to policy learning under partial identification of policy parameters (e.g., \citet{russell2020policy}, \citet{dadamo2022orthogonal}, \citet{christensen2023optimal}, \citet{yata2023optimal}).

\bibliographystyle{ecta}
\bibliography{policy}

\newpage
\appendix
\renewcommand{\thesection}{Appendix \Alph{section}}
\section{Proofs}
\label{proofs}
\renewcommand{\theequation}{A.\arabic{equation}}
\setcounter{equation}{0}

For a binary encouragement rule $\pi$ described in Section \ref{binaryencouragement}, define an ideal version of the empirical welfare criterion and budget based on the true propensity score and observed conditional average outcome as
\begin{align*}
    \bar{W}_n(\pi)&=E_n[ \pi(X,Z)\cdot\{\mu_Y(X,p(X,\alpha_1(X,Z)))-\mu_Y(X,p(X,\alpha_0(X,Z)))\}],\\
    B_n(\pi)&=E_n[C(X,Z)\cdot\{\pi(X,Z)\cdot p(X,\alpha_1(X,Z))+(1-\pi(X,Z))\cdot p(X,\alpha_0(X,Z))\}].
\end{align*}

\begin{proof}[Proof of Theorem \ref{representation}]
By (\ref{selection}), one can write
\begin{align*}
   Y(D(\boldsymbol{\alpha}(X,Z)))&=Y(1)\cdot D(\boldsymbol{\alpha}(X,Z))+Y(0)\cdot (1-D(\boldsymbol{\alpha}(X,Z)))\\
   &=Y(1)\cdot 1\{p(X,\boldsymbol{\alpha}(X,Z))\geq U\}+Y(0)\cdot 1\{p(X,\boldsymbol{\alpha}(X,Z))<U\}\\
   &=Y(1)\cdot 1\{p(X,Z)\geq U\}+Y(0)\cdot 1\{p(X,Z)<U\}\\
   &\quad\ +(Y(1)-Y(0))\cdot (1\{p(X,\boldsymbol{\alpha}(X,Z))\geq U\}- 1\{p(X,Z)\geq U\})\\
   &=Y+(Y(1)-Y(0))\cdot (1\{p(X,\boldsymbol{\alpha}(X,Z))\geq U\}- 1\{p(X,Z)\geq U\}).
\end{align*}
Hence,
\begin{align*}
    W(\boldsymbol{\alpha})-E[Y]&=E[(Y(1)-Y(0))\cdot (1\{p(X,\boldsymbol{\alpha}(X,Z))\geq U\}- 1\{p(X,Z)\geq U\})]\\
    &=E[E[E[Y(1)-Y(0)|X,Z,U]\\
    &\quad\ \cdot(1\{p(X,\boldsymbol{\alpha}(X,Z))\geq U\}- 1\{p(X,Z)\geq U\})|X,Z]]\\
    &=E[\operatorname{MTE}(X,U)\cdot E[1\{p(X,\boldsymbol{\alpha}(X,Z))\geq U\}- 1\{p(X,Z)\geq U\}|X,Z]]\\
    &=E\Big[\int_0^1 \operatorname{MTE}(X,u)\cdot(1\{p(X,\boldsymbol{\alpha}(X,Z))\geq u\}- 1\{p(X,Z)\geq u\})\,du\Big],
\end{align*}
where the second equality follows from the law of iterated expectations, the third equality follows from Assumption \ref{MTErestrictions}(ii) and the definition of the MTE, and the fourth equality follows from $U|X,Z\sim \operatorname{Unif}[0,1]$.
\end{proof}

\begin{proof}[Proof of Theorem \ref{upperFEWM}]
One has for any $\pi'\in\Pi$,
\begin{align*}
    W(\pi')-W(\hat{\pi}_{\mathrm{FEWM}})&=\bar{W}(\pi')-\bar{W}(\hat{\pi}_{\mathrm{FEWM}})\\
    &=[\bar{W}(\pi')-\bar{W}_n(\pi')]+[\bar{W}_n(\pi')-\hat{W}_n(\pi')]+[\hat{W}_n(\pi')-\hat{W}_n(\hat{\pi}_{\mathrm{FEWM}})]\\
    &\quad\ +[\hat{W}_n(\hat{\pi}_{\mathrm{FEWM}})-\bar{W}_n(\hat{\pi}_{\mathrm{FEWM}})]+[\bar{W}_n(\hat{\pi}_{\mathrm{FEWM}})-\bar{W}(\hat{\pi}_{\mathrm{FEWM}})]\\
    &\leq [\bar{W}(\pi')-\bar{W}_n(\pi')]+[\bar{W}_n(\pi')-\hat{W}_n(\pi')]\\
    &\quad\ +[\hat{W}_n(\hat{\pi}_{\mathrm{FEWM}})-\bar{W}_n(\hat{\pi}_{\mathrm{FEWM}})]+[\bar{W}_n(\hat{\pi}_{\mathrm{FEWM}})-\bar{W}(\hat{\pi}_{\mathrm{FEWM}})]\\
    &\leq2\sup_{\pi\in\Pi}|\bar{W}_n(\pi)-\bar{W}(\pi)|+2\sup_{\pi\in\Pi}|\hat{W}_n(\pi)-\bar{W}_n(\pi)|,
\end{align*}
where the first inequality follows from $\hat{W}_n(\hat{\pi}_{\mathrm{FEWM}})\geq \hat{W}_n(\pi')$. Hence,
\begin{equation*}
    E[R(\hat{\pi}_{\mathrm{FEWM}})]\leq 2E\Big[\sup_{\pi\in\Pi}|\bar{W}_n(\pi)-\bar{W}(\pi)|\Big]+2E\Big[\sup_{\pi\in\Pi}|\hat{W}_n(\pi)-\bar{W}_n(\pi)|\Big].
\end{equation*}
First, I bound $E[\sup_{\pi\in\Pi}|\bar{W}_n(\pi)-\bar{W}(\pi)|]$. Define
\begin{equation*}
    \mathcal{F}=\Big\{f_\pi(x,z)=\pi(x,z)\int_{p(x,\alpha_0(x,z))}^{p(x,\alpha_1(x,z))}\operatorname{MTE}(x,u)\,du:\pi\in\Pi\Big\}.
\end{equation*}
By Assumption \ref{boundedvc} and Lemma A.1 of \citet{kitagawa2018supplement}, $\mathcal{F}$ is a VC-subgraph class of functions with uniform envelope $\bar{M}$ and VC-dimension less than or equal to $\operatorname{VC}(\Pi)$. Then, one can apply Lemma A.4 of \citet{kitagawa2018supplement} to obtain
\begin{equation}
E\Big[\sup_{\pi\in\Pi}|\bar{W}_n(\pi)-\bar{W}(\pi)|\Big]=
    E\Big[\sup_{f\in\mathcal{F}}|E_n[f(X,Z)]-E[f(X,Z)]|\Big]\leq C_1\bar{M}\sqrt{\frac{\operatorname{VC}(\Pi)}{n}},
    \label{Wtof}
\end{equation}
where $C_1$ is a universal constant. Next, I bound $E[\sup_{\pi\in\Pi}|\hat{W}_n(\pi)-\bar{W}_n(\pi)|]$. By the triangle inequality, for any $\pi\in\Pi$,
\begin{align}
|\hat{W}_n(\pi)-\bar{W}_n(\pi)|
    &\leq \sum_{d\in\{0,1\}} E_n[|\hat{\mu}_Y(X,\hat{p}(X,\alpha_d(X,Z)))-\mu_Y(X,\hat{p}(X,\alpha_d(X,Z)))|]\nonumber\\
    &\quad\ +\sum_{d\in\{0,1\}} E_n[|\mu_Y(X,\hat{p}(X,\alpha_d(X,Z)))-\mu_Y(X,p(X,\alpha_d(X,Z)))|].\label{WhatWbar}
\end{align}
Each summand in the first term on the right-hand side of (\ref{WhatWbar}) can be bounded as
\begin{equation*}
    E_n[|\hat{\mu}_Y(X,\hat{p}(X,\alpha_d(X,Z)))-\mu_Y(X,\hat{p}(X,\alpha_d(X,Z)))|] \leq E_n\Big[\sup_{u\in[0,1]}|\hat{\mu}_Y(X,u)-\mu_Y(X,u)|\Big].
\end{equation*}
Each summand in the second term on the right-hand side of (\ref{WhatWbar}) can be bounded as
\begin{eqnarray*}
   &&E_n[|\mu_Y(X,\hat{p}(X,\alpha_d(X,Z)))-\mu_Y(X,p(X,\alpha_d(X,Z)))|]\\
   &\leq&\sup_{(u,x)\in[0,1]\times\mathcal{X}}|\operatorname{MTE}(x,u)|\cdot
   E_n[|\hat{p}(X,\alpha_d(X,Z))-p(X,\alpha_d(X,Z))|].
\end{eqnarray*}
Since these bounds do not depend on $\pi$,
\begin{align}
    E\Big[\sup_{\pi\in\Pi}|\hat{W}_n(\pi)-\bar{W}_n(\pi)|\Big]
    &\leq E\Big[\sup_{u\in[0,1]}|\hat{\mu}_Y(X,u)-\mu_Y(X,u)|\Big]\nonumber\\
    &\quad\ +\sup_{(u,x)\in[0,1]\times\mathcal{X}}|\operatorname{MTE}(x,u)|\cdot
   E[|\hat{p}(X,\alpha_d(X,Z))-p(X,\alpha_d(X,Z))|]\nonumber\\
    &=O(\psi_n^{-1}\vee\phi_n^{-1}),
    \label{convergeWn}
\end{align}
where the equality follows from Assumption \ref{estimatorpMTE}.
\end{proof}

\begin{proof}[Proof of Theorem \ref{sampleanaopt}]
For any $\epsilon>0$,
\begin{align}
    \Pr(W(\pi^*_{\mathrm{B}})-W(\hat{\pi}_{\mathrm{BEWM}})>\epsilon)
    &=\Pr(\bar{W}(\pi^*_{\mathrm{B}})-\bar{W}(\hat{\pi}_{\mathrm{BEWM}})>\epsilon,\hat{B}_n(\pi^*_{\mathrm{B}})\leq \kappa)\nonumber\\
    &\quad\ +\Pr(\bar{W}(\pi^*_{\mathrm{B}})-\bar{W}(\hat{\pi}_{\mathrm{BEWM}})>\epsilon,\hat{B}_n(\pi^*_{\mathrm{B}})>\kappa).
    \label{violationW}
    \end{align}
Noting that $\hat{B}_n(\pi^*_{\mathrm{B}})\leq\kappa$ implies $\hat{W}_n(\hat{\pi}_{\mathrm{BEWM}})\geq\hat{W}_n(\pi^*_{\mathrm{B}})$, the first term on the right-hand side of (\ref{violationW}) can be bounded as
    \begin{align}
      &\quad\ \Pr(\bar{W}(\pi^*_{\mathrm{B}})-\bar{W}(\hat{\pi}_{\mathrm{BEWM}})>\epsilon,\hat{B}_n(\pi^*_{\mathrm{B}})\leq \kappa)\nonumber\\
      &\leq \Pr(\bar{W}(\pi^*_{\mathrm{B}})-\hat{W}_n(\pi^*_{\mathrm{B}})+\hat{W}_n(\hat{\pi}_{\mathrm{BEWM}})-\bar{W}(\hat{\pi}_{\mathrm{BEWM}})>\epsilon,\hat{B}_n(\pi^*_{\mathrm{B}})\leq \kappa)\nonumber\\
      &\leq\Pr\Big(2\sup_{\pi\in\Pi}|\bar{W}(\pi)-\hat{W}_n(\pi)|>\epsilon\Big).\label{P1}
    \end{align}
The second term on the right-hand side of (\ref{violationW}) can be bounded as
\begin{align}
    \Pr(\bar{W}(\pi^*_{\mathrm{B}})-\bar{W}(\hat{\pi}_{\mathrm{BEWM}})>\epsilon,\hat{B}_n(\pi^*_{\mathrm{B}})>\kappa)&\leq\Pr(\hat{B}_n(\pi^*_{\mathrm{B}})>\kappa)\nonumber\\
    &\leq\Pr(\hat{B}_n(\pi^*_{\mathrm{B}})>B(\pi^*_{\mathrm{B}}))\nonumber\\
    &\leq\sup_{\epsilon'>0}\Pr(\hat{B}_n(\pi^*_{\mathrm{B}})-B(\pi^*_{\mathrm{B}})>\epsilon')\nonumber\\
    &\leq\sup_{\epsilon'>0}\Pr\Big(\sup_{\pi\in\Pi}|B(\pi)-\hat{B}_n(\pi)|>\epsilon'\Big),
    \label{P2}
\end{align}
where the second inequality follows from $B(\pi^*_{\mathrm{B}})\leq \kappa$. Also, for any $\epsilon>0$,
\begin{align}
    \Pr(B(\hat{\pi}_{\mathrm{BEWM}})-\kappa>\epsilon)
    &\leq\Pr(B(\hat{\pi}_{\mathrm{BEWM}})-\hat{B}_n(\hat{\pi}_{\mathrm{BEWM}})>\epsilon)\nonumber\\
    &\leq\Pr\Big(\sup_{\pi\in\Pi}|B(\pi)-\hat{B}_n(\pi)|>\epsilon\Big),
    \label{violationB}
\end{align}
where the first inequality follows from $\hat{B}_n(\hat{\pi}_{\mathrm{BEWM}})\leq \kappa$. By the triangle inequality, for any $\pi\in\Pi$,
\begin{equation}
    |B(\pi)-\hat{B}_n(\pi)|\leq|B(\pi)-B_n(\pi)|+|B_n(\pi)-\hat{B}_n(\pi)|.\label{triangleB}
\end{equation}
For the first term on the right-hand side of (\ref{triangleB}), from the same argument as the proof of (\ref{Wtof}), one can apply Lemma A.4 of \citet{kitagawa2018supplement} to obtain
\begin{equation}
    E\Big[\sup_{\pi\in\Pi}|B_n(\pi)-B(\pi)|\Big]=O(n^{-1/2}).\label{convergeB}
\end{equation}
For the second term on the right-hand side of (\ref{triangleB}), note that
\begin{eqnarray*}
    &&\sup_{\pi\in\Pi}|B_n(\pi)-\hat{B}_n(\pi)|\\
    &\leq&
    E_n[C(X,Z)\cdot(|\hat{p}(X,\alpha_1(X,Z))-p(X,\alpha_1(X,Z))|+|\hat{p}(X,\alpha_0(X,Z))-p(X,\alpha_0(X,Z))|)]\\
    &\leq&\sup_{x,z}C(x,z)\sum_{d\in\{0,1\}}E_n[|\hat{p}(X,\alpha_d(X,Z))-p(X,\alpha_d(X,Z))|].
\end{eqnarray*}
Hence, by Assumption \ref{estimatorpMTE}(i),
\begin{equation}
    E\Big[\sup_{\pi\in\Pi}|B_n(\pi)-\hat{B}_n(\pi)|\Big]=O(\psi_n^{-1}).
    \label{convergeBn}
\end{equation}
Combining (\ref{Wtof}), (\ref{convergeWn})--(\ref{convergeBn}), and Markov's inequality gives the desired result.
\end{proof}

\section{Verification of Assumption \ref{identifyW}}
\label{verify}

\begin{proof}[Verification of Assumption \ref{identifyW} in Example \ref{examplenonpara}]
    Assumption \ref{identifyW}(i) is satisfied because $p(x,z)$ is point-identified over $\operatorname{Supp}(X,Z)$. Assumption \ref{identifyW}(ii) is satisfied because for every $x\in\mathcal{X}$, $\operatorname{MTE}(\cdot,x)$ is point-identified over $\mathcal{P}(x)$, and $[\min\mathcal{P}_{\boldsymbol{\alpha}}(x),\max\mathcal{P}_{\boldsymbol{\alpha}}(x)]\subset[\min\mathcal{P}(x),\max\mathcal{P}(x)]\allowbreak\subset\mathcal{P}(x)$, where $\mathcal{P}(x)$ denotes the support of $p(X,Z)$ conditional on $X=x$.
\end{proof}

\begin{proof}[Verification of Assumption \ref{identifyW} in Example \ref{examplesemipara}]
   Assumption \ref{identifyW}(i) is satisfied because $\gamma$ is point-identified and $\theta(z)$ is point-identified over $\mathcal{Z}$ under (E2-5). Let $V=DV_1+(1-D)V_0$. Under (E2-4), one can write the conditional mean of $Y$ given $(X,p(X,Z))$ as
\begin{equation*}
   E[Y|X=x,p(X,Z)=u]=uX^\top\beta_1+(1-u)X^\top\beta_0+E[V|p(X,Z)=u].
\end{equation*}
Under (E2-5), $\beta_1$ and $\beta_0$ are point-identified, and $u\mapsto E[V|p(X,Z)=u]$ is point-identified over $\operatorname{Supp}(p(X,Z))$. Then, Assumption \ref{identifyW}(ii) is satisfied because for any $(x,u)\in\mathcal{X}\times\operatorname{Supp}(p(X,Z))$, the MTE is point-identified as
\begin{equation*}
    \operatorname{MTE}(x,u)=x^\top(\beta_1-\beta_0)+\frac{\partial}{\partial u}E[V|p(X,Z)=u],
\end{equation*}
and (E2-1) and (E2-2) imply that for every $x\in\mathcal{X}$, $[\min\mathcal{P}_{\boldsymbol{\alpha}}(x),\max\mathcal{P}_{\boldsymbol{\alpha}}(x)]\subset\cup_{x\in\mathcal{X}}[\min\mathcal{P}_{\boldsymbol{\alpha}}(x),\allowbreak\max\mathcal{P}_{\boldsymbol{\alpha}}(x)]\subset\cup_{x\in\mathcal{X}}[\min\mathcal{P}(x),\max\mathcal{P}(x)]\subset\operatorname{Supp}(p(X,Z))$.
\end{proof}

\begin{proof}[Verification of Assumption \ref{identifyW} in Example \ref{examplepara}]
    Assumption \ref{identifyW}(i) is satisfied because $\gamma$ is point-identified under (E3-3). Assumption \ref{identifyW}(ii) is satisfied because under (E3-2),
    \begin{equation*}
    \operatorname{MTE}(x,u)=x^\top(\beta_1-\beta_0)+\sum_{j=2}^J j\eta_ju^{j-1},
\end{equation*}
and $\beta_1$, $\beta_0$, and $\eta_2,\dots,\eta_J$ are point-identified under (E3-3).
\end{proof}

\section{Encouragement Rules with a Binary Instrument}
\label{encouragementbinary}
\renewcommand{\theequation}{C.\arabic{equation}}
\renewcommand{\thetable}{C.\arabic{table}}
\setcounter{equation}{0}
\setcounter{table}{0}
\setcounter{subsection}{0}

In many applications, the instrument is binary by construction. For completeness, I present a discussion of cases in which a binary instrument is intervened upon. An encouragement rule is a mapping $\pi:\mathcal{X}\to\{0,1\}$ that determines the manipulated value of the instrument. For example, the encouragement could be eligibility for welfare programs such as the National Job Training Partnership Act (JTPA) or the 401(k) retirement program. Define the social welfare criterion as $W(\pi)=E[Y(D(\pi(X)))]$. Given a feasible class $\Pi$ of encouragement rules, the optimal encouragement rule solves $\max_{\pi\in\Pi} W(\pi)$. Under (\ref{selection}) and Assumption \ref{MTErestrictions}, $W(\pi)$ is identified as
\begin{equation*}
    W(\pi)=E[E[Y|X,Z=1]\cdot\pi(X)+E[Y|X,Z=0]\cdot(1-\pi(X))].
\end{equation*}
Therefore, the optimal encouragement rule is identified without observing $D$. Intuitively, the optimization problem amounts to finding the optimal treatment rule that assigns $Z$ rather than $D$, from an intention-to-treat perspective. In this case, Assumption \ref{MTErestrictions}(i) warrants unconfoundedness, and the analysis essentially follows the original EWM framework. 

To understand which subpopulation benefits from the encouragement rule, I impose further assumptions on the selection behavior:

\begin{enumerate}[label=\textbf{Assumption C.\arabic*},ref=C.\arabic*,itemindent=5\parindent,leftmargin=0pt]
\item \label{monotonep} 
(Increasing Propensity Score) For each $x\in\mathcal{X}$, $p(x,1)\geq p(x,0)$.
\end{enumerate}

Under (\ref{selection}) and Assumption \ref{monotonep}, one can partition the population into three compliance groups: always-takers, never-takers, and compliers. Let $\varrho_c(x)=\Pr(D(1)>D(0)|X=x)$ be the conditional population share of compliers and $\Delta_c(x)=E[Y(1)-Y(0)|X=x,D(1)>D(0)]$ be the CATE for compliers. It is worth noting that $W(\pi)$ depends on $\pi$ only through the counterfactual outcome of compliers:
\begin{align*}
    W(\pi)&=E[Y(D(0))]+E[(Y(1)-Y(0))(D(1)-D(0))\cdot\pi(X)]\\
    &=E[Y(D(0))]+E[\Delta_c(X)\varrho_c(X)\cdot\pi(X)].
\end{align*}
For always-takers and never-takers, no encouragement rule can alter their outcomes. 

I proceed to present two extensions analogous to those in Section \ref{extensions}. In \ref{multiplebinaryiv}, I allow for multiple instruments. In \ref{budgetbinary}, I incorporate budget constraints.

\subsection{Multiple Instruments}
\label{multiplebinaryiv}

Consider a setting in which there are $L$ instruments available. Let $Z_1$ be the binary instrument that can be manipulated, $Z_{-1}\in\mathcal{Z}_{-1}\subset\mathbb{R}^{L-1}$ be a vector of additional instruments, and $Z=(Z_1,Z_{-1}^\top)^\top$. An encouragement rule is a mapping $\pi:\mathcal{X}\times\mathcal{Z}_{-1}\to\{0,1\}$ that determines the manipulated value of $Z_1$ while leaving $Z_{-1}$ unchanged. Define the social welfare criterion as $W(\pi)=E[Y(D(\pi(X,Z_{-1}),Z_{-1}))]$. Under (\ref{marginalselection}) and Assumption \ref{multipleMTErestrictions}, $W(\pi)$ is identified as
\begin{equation*}
    W(\pi)=E[E[Y|X,Z_{-1},Z_1=1]\cdot\pi(X,Z_{-1})+E[Y|X,Z_{-1},Z_1=0]\cdot(1-\pi(X,Z_{-1}))].
\end{equation*}
The additional instruments $Z_{-1}$ are treated equivalently to covariates $X$. 

In the case of two binary instruments, to understand which subpopulation benefits from the encouragement rule, I impose further assumptions on the selection behavior:

\begin{enumerate}[label=\textbf{Assumption C.\arabic*},ref=C.\arabic*,itemindent=5\parindent,leftmargin=0pt]
\setcounter{enumi}{1}
\item \label{monotonepXZ} 
(Component-Wise Increasing Propensity Score) For each $x\in\mathcal{X}$ and $z_1,z_2\in\{0,1\}$, $p(x,1,z_2)\geq p(x,0,z_2)$ and $p(x,z_1,1)\geq p(x,z_1,0)$.
\end{enumerate}

Under (\ref{marginalselection}) and Assumption \ref{monotonepXZ}, one can partition the population into the six compliance groups presented in Table \ref{compliance}. Denote by $G\in\mathcal{G}=\{\text{nt},\text{at},\text{1c},\text{2c},\text{ec},\text{rc}\}$ the compliance group identity. For each $g\in\mathcal{G}$ and $z_2\in\{0,1\}$, let $\varrho_g(x,z_2)=\Pr(G=g|X=x,Z_2=z_2)$ be the conditional population share of compliance group $g$ and $\Delta_g(x,z_2)=E[Y(1)-Y(0)|X=x,Z_2=z_2,G=g]$ be the compliance-group-specific CATE. Then, the social welfare criterion can also be written as
\begin{align*}
W(\pi)&=E[Y(D(0,Z_2))]+E[(Y(1)-Y(0))(D(1,Z_2)-D(0,Z_2))\cdot\pi(X,Z_2)]\\
&=E[Y(D(0,Z_2))]\\
&\quad\ +E[E[(Y(1)-Y(0))(D(1,1)-D(0,1))|X,Z_2=1]\cdot\Pr(Z_2=1|X)\cdot\pi(X,1)]\\
&\quad\ +E[E[(Y(1)-Y(0))(D(1,0)-D(0,0))|X,Z_2=0]\cdot\Pr(Z_2=0|X)\cdot\pi(X,0)]\\
&=E[Y(D(0,Z_2))]\\
&\quad\ +E[(\varrho_{\mathrm{1c}}(X,1)\Delta_{\mathrm{1c}}(X,1)+\varrho_{\mathrm{rc}}(X,1)\Delta_{\mathrm{rc}}(X,1))\cdot\Pr(Z_2=1|X)\cdot\pi(X,1)]\\
&\quad\ +E[(\varrho_{\mathrm{1c}}(X,0)\Delta_{\mathrm{1c}}(X,0)+\varrho_{\mathrm{ec}}(X,0)\Delta_{\mathrm{ec}}(X,0))\cdot\Pr(Z_2=0|X)\cdot\pi(X,0)].
\end{align*}

Therefore, besides $Z_1$ compliers, the encouragement rule would affect the outcomes of reluctant compliers with $Z_2=1$ and eager compliers with $Z_2=0$.

\begin{table}[ht]
\centering
\begin{threeparttable}
\caption{Compliance Groups \citep[][Proposition 4]{mogstad2021causal}}
\begin{tabular}{l| c | c | c | c }
\hline
Name & $D(0,0)$ & $D(0,1)$ & $D(1,0)$ & $D(1,1)$\\
\hline
Never-takers (nt) & N & N & N & N\\
Always-takers (at) & T & T & T & T\\
$Z_1$ compliers (1c) & N & N & T & T\\
$Z_2$ compliers (2c) & N & T & N & T\\
Eager compliers (ec) & N & T & T & T\\
Reluctant compliers (rc)\hspace{20pt} & N & N & N & T\\
\hline
\end{tabular}
\begin{tablenotes}
\item Notes: ``T'' indicates treatment, and ``N'' indicates non-treatment.
\end{tablenotes}
\label{compliance}
\end{threeparttable}
\end{table}

\subsection{Budget Constraints}
\label{budgetbinary}

Suppose that the policymaker faces a harsh budget constraint such that the proportion of the population receiving treatment cannot exceed $\kappa\in(0,1)$. Then she faces a constrained optimization problem:
\begin{equation*}
    \max_{\pi\in\Pi}W(\pi) \text{ s.t. } B(\pi)\leq \kappa,
\end{equation*}
where $B(\pi)=E[D(\pi(X))]$. Let $\pi_B^*$ denote the solution. I maintain (\ref{selection}) and Assumptions \ref{MTErestrictions} and \ref{monotonep} for the rest of this subsection.\footnote{I implicitly assume that $E[D(0)]\leq\kappa$.}

When $\Pi$ is unrestricted, $\pi_B^*$ coincides with the optimal deterministic \emph{individualized encouragement rule} studied by \citet{qiu2020optimal}. Provided that the distribution of $\Delta_c(X)$ is continuous, $\pi_B^*$ admits a closed form: $\pi_B^*(x)=1\{\Delta_c(x)\geq\max\{\underline{\Delta},0\}\}$, where $\underline{\Delta}=\inf\{\Delta\in\mathbb{R}:E[1\{\Delta_c(X)\geq\Delta\}\varrho_c(X)]\leq\kappa-E[D(0)]\}$. Intuitively, $\pi_B^*$ assigns encouragement to compliers in decreasing order of $\Delta_c(x)$ until the resources are exhausted.

Alternatively, one simple idea is to enforce random rationing as in \citet{kitagawa2018should}: if $\pi$ violates the resource constraint, then the encouragement is randomly allocated to a fraction $\frac{\kappa-E[D(0)]}{E[D(\pi(X))]-E[D(0)]}$ of individuals with $\pi(X)=1$ independently of everything else. Define the resource-constrained welfare criterion as
\begin{equation*}
W^\kappa(\pi)=E[Y(D(0))]+E[(Y(D(1))-Y(D(0)))\cdot\pi(X)]\cdot\min\Big\{1,\frac{\kappa-E[D(0)]}{E[D(\pi(X))]-E[D(0)]}\Big\}.
\end{equation*}
The optimal encouragement rule under random rationing solves $\max_{\pi\in\Pi}W^\kappa(\pi)$, which uses the resources less efficiently than $\pi_B^*$.

\section{Sufficient Conditions for Assumption \ref{estimatorpMTE}(i)}
\renewcommand{\theequation}{D.\arabic{equation}}
\renewcommand{\thelemma}{D.\arabic{lemma}}
\renewcommand{\theprop}{D.\arabic{prop}}
\renewcommand{\thetheorem}{D.\arabic{theorem}}
\setcounter{equation}{0}
\setcounter{theorem}{0}
\label{supnorm}

Consider the nonparametric regression model
\begin{equation*}
    D=p(\tilde{X})+\epsilon,\quad E[\epsilon|\tilde{X}]=0,
\end{equation*}
where $\tilde{X}=(X^\top,Z)^\top\in\tilde{\mathcal{X}}\subset\mathbb{R}^{d_{\tilde{x}}}$. To force the resulting estimator $\hat{p}(\tilde{x})$ to lie between 0 and 1, one can use a trimmed version as in \citet[Eq. (4.2)]{carneiro2009estimating}.

\subsection{Local Polynomial Estimators}
\label{localpoly}

I follow \citet{audibert2007fast} to impose the following restrictions:

\begin{enumerate}[label=\textbf{Assumption D.\arabic*},ref=D.\arabic*,itemindent=5\parindent,leftmargin=0pt]
\item \label{localpolysass} 
(Local Polynomial Estimators)  
\begin{enumerate}[label=(\roman*)]
\item \label{holder}
$p(\cdot)$ belongs to a H\"{o}lder class of functions with degree $s\geq1$ and constant $0<R<\infty$.
\item \label{px}
Let $\operatorname{Leb}(\cdot)$ be the Lebesgue measure on $\mathbb{R}^{d_{\tilde{x}}}$. There exist constants $c_0$ and $r_0$ such that
\begin{equation*}
    \operatorname{Leb}(\tilde{\mathcal{X}}\cap B(\tilde{x},r))\geq c_0\operatorname{Leb}(B(\tilde{x},r))\quad \forall 0<r\leq r_0,\ \forall \tilde{x}\in\tilde{\mathcal{X}},
\end{equation*}
where $B(\tilde{x},r)$ denotes the closed Euclidean ball in $\mathbb{R}^{d_{\tilde{x}}}$ centered at $\tilde{x}$ and of radius $r$. Moreover, the marginal distribution of $\tilde{X}$ has the density function $f$ with respect to the Lebesgue measure of $\mathbb{R}^{d_{\tilde{x}}}$ such that $0<f_{\min}\leq f(\tilde{x})\leq f_{\max}<\infty$ $\forall\tilde{x}\in\tilde{\mathcal{X}}$.
\item \label{kernel}
The kernel function $K(\cdot)$ satisfies: $\exists c>0$ such that $K(u)\geq c1\{\|u\|\leq c\}$ $\forall u\in\mathbb{R}^{d_{\tilde{x}}}$, $\int_{\mathbb{R}^{d_{\tilde{x}}}}K(u)du=1$, $\int_{\mathbb{R}^{d_{\tilde{x}}}}(1+\|u\|^{4s})K^2(u)du<\infty$, and $\sup_{u\in\mathbb{R}^{d_{\tilde{x}}}} (1+\|u\|^{2s})K(u)<\infty$.
\end{enumerate}
\end{enumerate}

I consider the local polynomial regression fit for $p(\tilde{x})$ with degree $l=s-1$:
\begin{align}
    \hat{p}(\tilde{x})&=(\hat{\xi}(\tilde{x}))^\top U(0)\cdot 1\{\lambda_{\mathrm{min}}(B(\tilde{x}))\geq (\log n)^{-1}\},\label{localpolydef}\\
    \hat{\xi}(\tilde{x})&=\argmin_{\xi}\sum_{i=1}^n\Big[D_i-\xi^\top U\Big(\frac{\tilde{X}_i-\tilde{x}}{h}\Big)\Big]^2K\Big(\frac{\tilde{X}_i-\tilde{x}}{h}\Big),\nonumber
\end{align}
where $U\Big(\frac{\tilde{X}_i-\tilde{x}}{h}\Big)$ is a vector with elements indexed by the multi-index $(\ell_1,\dots,\ell_{d_{\tilde{x}}})\in\mathbb{N}^{d_{\tilde{x}}}$: 
\begin{equation*}
    U\Big(\frac{\tilde{X}_i-\tilde{x}}{h}\Big)=\Big\{\Big(\frac{\tilde{X}_i-\tilde{x}}{h}\Big)_1^{\ell_1}\dots\Big(\frac{\tilde{X}_i-\tilde{x}}{h}\Big)_{d_{\tilde{x}}}^{\ell_{d_{\tilde{x}}}}\Big\}_{0\leq \sum_{i=1}^{d_{\tilde{x}}}\ell_i\leq l},
\end{equation*}
and $\lambda_{\mathrm{min}}(B(\tilde{x}))$ is the smallest eigenvalue of $B(\tilde{x})=\frac{1}{nh^{d_{\tilde{x}}}}\sum_{i=1}^n U\Big(\frac{\tilde{X}_i-\tilde{x}}{h}\Big)\Big(U\Big(\frac{\tilde{X}_i-\tilde{x}}{h}\Big)\Big)^\top K\Big(\frac{\tilde{X}_i-\tilde{x}}{h}\Big)$. Theorem \ref{localpolyrate} states the sup-norm convergence rate in expectation for $\hat{p}$. It is straightforward to calculate the fastest convergence rate in Assumption \ref{estimatorpMTE}(i) as $\psi_n=n^{-\frac{1}{2+d_{\tilde{x}}/s}}$.

\begin{theorem}\label{localpolyrate}
Suppose that Assumption \ref{localpolysass} holds. Then, for the local polynomial fit $\hat{p}(\cdot)$ for $p(\cdot)$ defined in (\ref{localpolydef}),
    \begin{equation*}
    E\Big[\sup_{\tilde{x}\in\tilde{\mathcal{X}}}|\hat{p}(\tilde{x})-p(\tilde{x})|\Big]=O\Big(h^s+\frac{1}{\sqrt{nh^{d_{\tilde{x}}}}}\Big).
\end{equation*}
\end{theorem}

\begin{proof}
    By Theorem 3.2 of \citet{audibert2007fast}, there exist positive constants $C_1,C_2,C_3$ that only depend only on $s$, $R$, $d_{\tilde{x}}$, $c_0$, $r_0$, $f_{\min}$, $f_{\max}$, and the kernel $K$, such that, for any $0<h<r_0/c_0$, and $C_3h^s<\delta$, and any $n\geq1$,
    \begin{equation*}
        \Pr\Big(\sup_{\tilde{x}\in\tilde{\mathcal{X}}}|\hat{p}(\tilde{x})-p(\tilde{x})|\geq\delta\Big)\leq C_1\exp(-C_2nh^{d_{\tilde{x}}}\delta^2).
    \end{equation*}
    It follows that 
    \begin{align*}
        E\Big[\sup_{\tilde{x}\in\tilde{\mathcal{X}}}|\hat{p}(\tilde{x})-p(\tilde{x})|\Big]&=\int_0^\infty \Pr\Big(\sup_{\tilde{x}\in\tilde{\mathcal{X}}}|\hat{p}(\tilde{x})-p(\tilde{x})|\geq\delta\Big)d\delta\\
        &\leq C_3h^s+C_1\int_0^\infty \exp(-C_2nh^{d_{\tilde{x}}}\delta^2)d\delta\\
        &=C_3h^s+\frac{C_4}{\sqrt{nh^{d_{\tilde{x}}}}},
    \end{align*}
    where $C_4=\frac{C_1}{2}\sqrt{\frac{\pi}{C_2}}$.
\end{proof}

\subsection{Series Estimators}
\label{series}

Let me introduce some notations. Let $\lambda_{\mathrm{min}}(\cdot)$ denote the smallest eigenvalue of a matrix. Let the exponent $^-$ denote the Moore--Penrose generalized inverse. Let $\|\cdot\|$ denote the Euclidean norm when applied to vectors and the matrix spectral norm (i.e., the largest singular value) when applied to matrices. Let $\|\cdot\|_\infty$ denote the sup norm, i.e., if $f:\tilde{\mathcal{X}}\to\mathbb{R}$ then $\|f\|_\infty=\sup_{\tilde{x}\in\tilde{\mathcal{X}}}|f(\tilde{x})|$.

I consider the series least-squares estimator of $p$:
\begin{equation*}
    \hat{p}(\tilde{x})=b^k(\tilde{x})^\top(B^\top B)^-B^\top\mathbf{D},
\end{equation*}
where $b_{k1},\dots,b_{kk}$ are a collection of $k$ sieve basis functions, and
\begin{equation*}
    b^k(\tilde{x})=(b_{k1}(\tilde{x}),\dots,b_{kk}(\tilde{x}))^\top,\quad
    B=(b^k(\tilde{X}_1),\dots,b^k(\tilde{X}_n))^\top,\quad
    \mathbf{D}=(D_1,\dots,D_n)^\top.
\end{equation*}
Define $\zeta_k=\sup_{\tilde{x}\in\tilde{\mathcal{X}}}\|b^k(\tilde{x})\|$ and $\lambda_k=[\lambda_{\mathrm{min}}(E[b^k(\tilde{X})b^k(\tilde{X})^\top])]^{-1/2}$. Let $B_k=\operatorname{clsp}\{b_{k1},\dots,b_{kk}\}$ denote the closed linear span of the basis functions. I impose the following regularity conditions:

\begin{enumerate}[label=\textbf{Assumption D.\arabic*},ref=D.\arabic*,itemindent=5\parindent,leftmargin=0pt]
\setcounter{enumi}{1}
\item \label{seriessass} 
(Series Estimators)  
\begin{enumerate}[label=(\roman*)]
\item \label{lambdaminEbb} 
$\lambda_{\mathrm{min}}(E[b^k(\tilde{X})b^k(\tilde{X})^\top])>0$ for each $k$.
\item \label{epsilon} 
There exist $\nu,\tilde{c}>0$ such that $\sup_{\tilde{x}\in\tilde{\mathcal{X}}}E[|\epsilon|^q|\tilde{X}=\tilde{x}]\leq\frac{q!}{2}\nu^2\tilde{c}^{q-2}$ for all $q\geq 2$.
\item \label{approxrate} 
There exists $\rho>0$ such that $\inf_{h\in B_k}\|p-h\|_\infty=O(k^{-\rho})$ as $k\to\infty$.
\end{enumerate}
\end{enumerate}

Let $\tilde{b}^k(\tilde{x})$ denote the orthonormalized vector of basis functions, namely
\begin{equation*}
    \tilde{b}^k(\tilde{x})=E[b^k(\tilde{X})b^k(\tilde{X})^\top]^{-1/2}b^k(\tilde{x})
\end{equation*}
and let $\tilde{B}=(\tilde{b}^k(\tilde{X}_1),\dots,\tilde{b}^k(\tilde{X}_n))^\top$. Let $\tilde{p}$ denote the projection of $p$ onto $B_k$ under the empirical measure, that is,
\begin{equation*}
    \tilde{p}(\tilde{x})=b^k(\tilde{x})^\top(B^\top B)^-B^\top\mathbf{P}=\tilde{b}^k(\tilde{x})^\top(\tilde{B}^\top\tilde{B})^-\tilde{B}^\top\mathbf{P},
\end{equation*}
where $\mathbf{P}=(p(\tilde{X}_1),\dots,p(\tilde{X}_n))^\top$. One can bound $\|\hat{p}-p\|_\infty$ using
\begin{align*}
    \|\hat{p}-p\|_\infty&\leq\|p-\tilde{p}\|_\infty+\|\hat{p}-\tilde{p}\|_\infty\\
    &=\text{bias term}+\text{variance term}.
\end{align*}

I state three preparatory lemmas. Lemma \ref{BBnI} gives an exponential tail bound for $\|\tilde{B}^\top\tilde{B}/n-I_k\|$. Lemma \ref{varianceterm} provides an exponential tail bound on the sup-norm variance term. Lemma \ref{biasterm} provides a bound on the sup-norm bias term. The proofs are in \ref{proofauxiliary}.

\begin{lemma}\label{BBnI}
Suppose that Assumption \ref{seriessass}(i) holds. Then,
\begin{equation*}
    \Pr(\|\tilde{B}^\top\tilde{B}/n-I_k\|\geq \frac{1}{2})\leq 2k\exp\Big(-\frac{C_5n}{\zeta_k^2\lambda_k^2+1}\Big)
\end{equation*}
for some finite positive constant $C_5$.
\end{lemma}

Let $\mathcal{A}_n$ denote the event on which $\|\tilde{B}^\top\tilde{B}/n-I_k\|\leq \frac{1}{2}$. Let $1_{\mathcal{A}_n}$ denote the indicator function of $\mathcal{A}_n$.

\begin{lemma}\label{varianceterm}
Suppose that Assumptions \ref{seriessass}(i) and (ii) hold. Then, for any $\delta\in(0,1]$,
\begin{equation*}
    \Pr(1_{\mathcal{A}_n}\|\hat{p}-\tilde{p}\|_\infty\geq\delta)\leq\exp\Big(-\frac{C_6n\delta^2}{\zeta_k^2\lambda_k^2}\Big)
\end{equation*}
for some finite positive constant $C_6$.
\end{lemma}

Let $P_{k,n}$ be the empirical projection operator onto $B_k$, namely
\begin{equation*}
    P_{k,n}h(\tilde{x})=b^k(\tilde{x})^\top(B^\top B)^-B^\top\mathbf{H}=\tilde{b}^k(\tilde{x})^\top(\tilde{B}^\top\tilde{B})^-\tilde{B}^\top\mathbf{H}
\end{equation*}
for any $h:\tilde{\mathcal{X}}\to\mathbb{R}$, where $\mathbf{H}=(h(\tilde{X}_1),\dots,h(\tilde{X}_n))^\top$. Let $L^\infty(\tilde{X})$ denote the space of bounded functions under the sup norm. Let
\begin{equation*}
    \|P_{k,n}\|_\infty=\sup_{h\in L^\infty(\tilde{X}):\|h\|_\infty\neq0}\frac{\|P_{k,n}h\|_\infty}{\|h\|_\infty}.
\end{equation*}

\begin{lemma}\label{biasterm}
Suppose that Assumption \ref{seriessass}(i) holds. Then,
\begin{equation*}
   \|\tilde{p}-p\|_\infty\leq(1+\|P_{k,n}\|_\infty)\inf_{h\in B_k}\|p-h\|_\infty
\end{equation*}
and $1_{\mathcal{A}_n}\|P_{k,n}\|_\infty\leq \sqrt{2}\zeta_k\lambda_k$.
\end{lemma}

Theorem \ref{supnormconv} concludes the sup-norm convergence rate in expectation for $\hat{p}$.

\begin{theorem}\label{supnormconv}
Suppose that Assumption \ref{seriessass} holds. Then,
\begin{equation*}
E[\|\hat{p}-p\|_\infty]=O(\zeta_k\lambda_k(k^{-\rho}+k/\sqrt{n})).
\end{equation*}
\end{theorem}

\begin{proof}
Combining Lemmas \ref{BBnI}--\ref{biasterm}, there exist positive constants $C_7$, $C_8$, and $C_9$ such that if $C_7\zeta_k\lambda_kk^{-\rho}\leq\delta\leq1$, 
\begin{align*}
    \Pr(\|\hat{p}-p\|_\infty\geq\delta)&\leq 1-\Pr(\mathcal{A}_n)+\Pr(1_{\mathcal{A}_n}\|\hat{p}-p\|_\infty\geq\delta)\\
    &\leq \Pr(\|\tilde{B}^\top\tilde{B}/n-I_k\|\geq \frac{1}{2})+\Pr(1_{\mathcal{A}_n}\|\hat{p}-\tilde{p}\|_\infty\geq\frac{\delta}{2})\\
    &\leq C_8k\exp\Big(-\frac{C_9n\delta^2}{\zeta_k^2\lambda_k^2}\Big)
\end{align*}
for all sufficiently large $k$. For $\delta>1$, the same inequality holds since $\hat{p}$ and $p$ take values in $[0,1]$. It follows that
\begin{align*}
    E[\|\hat{p}-p\|_\infty]&=\int_0^\infty \Pr(\|\hat{p}-p\|_\infty\geq\delta)d\delta\\
    &\leq C_7\zeta_k\lambda_kk^{-\rho}+C_8k\int_0^\infty\exp\Big(-\frac{C_9n\delta^2}{\zeta_k^2\lambda_k^2}\Big)d\delta\\
    &=C_7\zeta_k\lambda_kk^{-\rho}+C_{10}\zeta_k\lambda_k\frac{k}{\sqrt{n}}
\end{align*}
for all sufficiently large $k$, where $C_{10}=\frac{C_8}{2}\sqrt{\frac{\pi}{C_9}}$.
\end{proof}

When $\tilde{\mathcal{X}}$ is compact and rectangular, and $X$ has a probability density function that is bounded away from zero on $\tilde{\mathcal{X}}$, $\zeta_k\lambda_k=O(\sqrt{k})$ for regression splines and $\zeta_k\lambda_k=O(k)$ for polynomials. If one further assumes that $p(\cdot)$ is continuously differentiable of order $s$ on $\tilde{\mathcal{X}}$, then Assumption \ref{seriessass}(iii) holds with $\rho=s/d_{\tilde{x}}$. It is straightforward to calculate the fastest convergence rate in Assumption \ref{estimatorpMTE}(i) as $\psi_n=n^{\frac{1}{2}-\frac{3}{4}\cdot\frac{1}{1+s/d_{\tilde{x}}}}$ for regression splines and $\psi_n=n^{\frac{1}{2}-\frac{1}{1+s/d_{\tilde{x}}}}$ for polynomials.

\section{Sufficient Conditions for Assumption \ref{estimatorpMTE}(ii)}
\renewcommand{\theequation}{E.\arabic{equation}}
\renewcommand{\theprop}{E.\arabic{prop}}
\setcounter{equation}{0}
\label{parametricMTE}

Consider a parametric regression function for the observed conditional average outcome:
\begin{equation*}
    \mu_Y(x,u)=x^\top\beta_0+x^\top(\beta_1-\beta_0)u+\sum_{j=2}^J\eta_j u^j.
\end{equation*}
Let $W=((1-p(X,Z))X^\top,p(X,Z)X^\top,p(X,Z)^2,\dots,p(X,Z)^J)^\top$ and $\vartheta=(\beta_0^\top,\beta_1^\top,\eta_2,\dots,\eta_J)^\top$. Then, $\vartheta$ is identified as $\vartheta=E[WW^\top]^{-1}E[WY]$ provided that $E[WW^\top]$ is positive definite. Given an estimator $\hat{p}(x,z)$ for $p(x,z)$, define the regressor as
\begin{equation*}
    \hat{W}=((1-\hat{p}(X,Z))X^\top,\hat{p}(X,Z)X^\top,\hat{p}(X,Z)^2,\dots,\hat{p}(X,Z)^J)^\top.
\end{equation*}
The OLS estimator for $\vartheta$ is obtained by regressing $Y$ on $\hat{W}$:
\begin{equation*}
    \hat{\vartheta}=(\hat{\beta}_0^\top,\hat{\beta}_1^\top,\hat{\eta}_2,\dots,\hat{\eta}_J)^\top=E_n[\hat{W}\hat{W}^\top]^{-1}E_n[\hat{W}Y].
\end{equation*}
Then, $\mu_Y(x,u)$ can be estimated by
\begin{equation*}
    \hat{\mu}_Y(x,u)=x^\top\hat{\beta}_0+x^\top(\hat{\beta}_1-\hat{\beta}_0)u+\sum_{j=2}^J\hat{\eta}_j u^j.
\end{equation*}

For this concrete estimator, I provide primitive conditions that guarantee the high-level condition in Assumption \ref{estimatorpMTE}(ii) to hold.

\begin{enumerate}[label=\textbf{Assumption E.\arabic*},ref=E.\arabic*,itemindent=5\parindent,leftmargin=0pt]
\item \label{polyMTE} 
(Polynomial MTE Model) Let $C$ and $c$ be positive constants.
\begin{enumerate}[label=(\roman*)]
\item \label{boundedsupport}
$X$ and $Y$ have bounded support.
\item \label{compactTheta}
The parameter space for $\vartheta$ is compact so that for sufficiently large n, $\|\hat{\vartheta}\|+\|\vartheta\|\leq C$ almost surely.
\item \label{lambdaminEWW}
$\lambda_{\min}(E[WW^\top])\geq c$, where $\lambda_{\min}(\cdot)$ denotes the smallest eigenvalue of a matrix.
\end{enumerate}
\end{enumerate}

\begin{prop}\label{parametricMTEcond}
Suppose that Assumptions \ref{estimatorpMTE}(i) and \ref{polyMTE} hold. Then, Assumption \ref{estimatorpMTE}(ii) holds with $\phi_n=\psi_n$.
\end{prop}

\begin{proof}
One has
\begin{align*}
    &E\Big[\sup_{u\in[0,1]}|\hat{\mu}_Y(X,u)-\mu_Y(X,u)|\Big]\\
    &=E\Big[\sup_{u\in[0,1]}\Big|X^\top[\hat{\beta}_0+(\hat{\beta}_1-\hat{\beta}_0)u-\beta_0-(\beta_1-\beta_0)u]+\sum_{j=2}^J (\hat{\eta}_j-\eta_j) u^j\Big|\Big]\\
    &\leq \sum_{d\in\{0,1\}}E[|X^\top(\hat{\beta}_d-\beta_d)|]+\sum_{j=2}^J E[|\hat{\eta}_j-\eta_j|],
\end{align*}
where the inequality follows from the triangle inequality. Given Assumption \ref{polyMTE}(i), it suffices to show
\begin{equation*}
    E[\|\hat{\vartheta}-\vartheta\|]=O(\psi_n^{-1}).
\end{equation*}
One can write
\begin{align*}
    E[WW^\top](\hat{\vartheta}-\vartheta)&=E_n[(\hat{W}-W)Y]-E_n[\hat{W}\hat{W}^\top-WW^\top]\hat{\vartheta}\\
    &\quad\ +(E_n-E)[WY]-(E_n-E)[WW^\top]\hat{\vartheta}.
\end{align*}
Let $\|\cdot\|$ denote the matrix spectral norm when applied to matrices. There exist positive constants $C_1,C_2<\infty$ such that
\begin{align*}
    \|(\hat{W}-W)Y\|&\leq\|\hat{W}-W\||Y|\leq|Y|(C_1+C_2\|X\|)|\hat{p}(X,Z)-p(X,Z)|,\\
    \|\hat{W}\hat{W}^\top-WW^\top\|&\leq\|\hat{W}-W\|\|\hat{W}+W\|\leq(C_1+C_2\|X\|^2)|\hat{p}(X,Z)-p(X,Z)|.
\end{align*}
Then, by Assumptions \ref{estimatorpMTE}(ii) and \ref{polyMTE}(i),
\begin{align}
    E[E_n[\|(\hat{W}-W)Y\|]]&=O(\psi_n^{-1}),\label{vartheta1}\\
    E[E_n[\|\hat{W}\hat{W}^\top-WW^\top\|]]&=O(\psi_n^{-1}).\label{vartheta2}
\end{align}
Also, by Assumption \ref{polyMTE}(i),
\begin{align}
    E[\|(E_n-E)[WY]\|]&=O(n^{-1/2}),\label{vartheta3}\\
    E[\|(E_n-E)[WW^\top]\|]&=O(n^{-1/2}).\label{vartheta4}
\end{align}
By the triangle inequality,
\begin{align*}
    E[\|\hat{\vartheta}-\vartheta\|]&\leq \lambda_{\min}(E[WW^\top])^{-1}E[E_n[\|(\hat{W}-W)Y\|]]\\
    &\quad\ +\lambda_{\min}(E[WW^\top])^{-1}E[E_n[\|\hat{W}\hat{W}^\top-WW^\top\|]\|\hat{\vartheta}\|]\\
    &\quad\ +\lambda_{\min}(E[WW^\top])^{-1}E[\|(E_n-E)[WY]\|]\\
    &\quad\ +\lambda_{\min}(E[WW^\top])^{-1}E[\|(E_n-E)[WW^\top]\|\|\hat{\vartheta}\|],
\end{align*}
where the right-hand side is $O(\psi_n^{-1})$ by Assumptions \ref{polyMTE}(ii)--(iii) and (\ref{vartheta1})--(\ref{vartheta4}).
\end{proof}

\section{Doubly Robust Approach}
\renewcommand{\theequation}{F.\arabic{equation}}
\renewcommand{\thelemma}{F.\arabic{lemma}}
\renewcommand{\thetheorem}{F.\arabic{theorem}}
\setcounter{equation}{0}
\setcounter{theorem}{0}
\setcounter{lemma}{0}
\label{doublyrobust}

A natural candidate for the doubly robust score arises from studying the influence function of the social welfare criterion. As discussed in Sections \ref{secidentification} and \ref{binaryencouragement}, under Assumptions \ref{MTErestrictions} and \ref{identifyW}, one has
\begin{equation*}
    W(\pi)=E[Y(D(\alpha_0(X,Z)))]+E[\pi(X,Z)\cdot(\varphi(X,p(X,\alpha_1(X,Z));p)-\varphi(X,p(Z,\alpha_0(X,Z));p))],
\end{equation*}
where $\varphi(x,u;p)=E[Y|X=x,p(X,Z)=u]$. Here I use a different notation for $E[Y|X=x,p(X,Z)=u]$ than in the main text to emphasize that $p$ is unknown and needs to be estimated. Fix $\pi\in\Pi$ and define
\begin{equation*}
    m(x,z;p,\varphi)=\pi(x,z)(\varphi(x,p(x,\alpha_1(x,z));p)-\varphi(x,p(x,\alpha_0(x,z));p))
\end{equation*}
so that $W(\pi)=E[Y(D(\alpha_0(X,Z)))]+E[m(X,Z;p,\varphi)]$. Lemma \ref{influence} calculates the influence function of $E[m(X,Z;p,\varphi)]$ following \citet{ichimura2022influence}. A proof is provided in \ref{proofauxiliary}. It is worth noting that the estimation error of $p$ has two contributions to the influence function of $E[m(X,Z;p,\varphi)]$ through $\varphi$. First, the estimate of $p$ serves as a generated regressor that changes the conditional expectation estimator. Second, the estimate of $p$ enters as the argument at which the conditional expectation estimator is evaluated. I adopt the approach from \citet{hahn2013asymptotic} to show that the two contributions cancel each other.

\begin{lemma}\label{influence}
Let $F_\tau=(1-\tau)F_0+\tau H,0<\tau<1$ denote a convex combination of the true CDF $F_0$ with another CDF $H$. Let $p_\tau=p(F_\tau)$ and $\varphi_\tau=\varphi(F_\tau)$. Then,
\begin{equation*}
    \frac{\partial}{\partial\tau}E[m(X,Z;p_\tau,\varphi_\tau)]=E_H[\pi(X,Z)g(X,Z)(Y-\varphi_0(X,p_0(X,Z);p_0))],
\end{equation*}
where
\begin{equation*}
g(x,z)=\frac{f_{X,\alpha_1(X,Z)}(x,z)-f_{X,\alpha_0(X,Z)}(x,z)}{f_{X,Z}(x,z)}.
\end{equation*}
\end{lemma}

In light of Lemma \ref{influence}, one can construct the doubly robust score in the following steps.

\begin{enumerate}
    \item Divide the data into $K$ evenly-sized folds. For each fold $k=1,\dots,K$, use the other $K-1$ data folds to 
    \begin{enumerate}
        \item estimate $p(x,z)$ and $g(x,z)$; denote the resulting estimates by $\hat{p}^{(-k)}(x,z)$ and $\hat{g}^{(-k)}(x,z)$;
        \item estimate $\varphi(x,u;\hat{p}^{(-k)})$; denote the resulting estimate by $\hat{\varphi}^{(-k)}(x,u;\hat{p}^{(-k)})$.
    \end{enumerate} 
    \item Calculate the doubly robust score as
    \begin{align*}
    \hat{\Gamma}_i&=\hat{\varphi}^{(-k(i))}(X_i,\hat{p}^{(-k(i))}(X_i,\alpha_1(X_i,Z_i));\hat{p}^{(-k(i))})\\
    &\quad\ -\hat{\varphi}^{(-k(i))}(X_i,\hat{p}^{(-k(i))}(X_i,\alpha_0(X_i,Z_i));\hat{p}^{(-k(i))})\\
    &\quad\ +\hat{g}^{(-k(i))}(X_i,Z_i)(Y_i-\hat{\varphi}^{(-k(i))}(X_i,\hat{p}^{(-k(i))}(X_i,Z_i);\hat{p}^{(-k(i))})),
    \end{align*}
where $k(i)\in\{1,\dots,K\}$ denotes the fold containing the $i$th observation. 
\end{enumerate}

Define the \textit{doubly-robust encouragement rule} as
\begin{equation}
    \hat{\pi}_{\mathrm{DR}}\in\argmax_{\pi\in\Pi}\hat{W}_n^{\mathrm{DR}}(\pi),\quad \hat{W}_n^{\mathrm{DR}}(\pi)=\frac{1}{n}\sum_{i=1}^n \pi(X_i,Z_i)\hat{\Gamma}_i.
    \label{DRobjective}
\end{equation}
To analyze the regret of $\hat{\pi}_{\mathrm{DR}}$, I impose the following assumptions:

\begin{enumerate}[label=\textbf{Assumption F.\arabic*},ref=F.\arabic*,itemindent=5\parindent,leftmargin=0pt]
\item \label{DRassumption} 
(Doubly-Robust Encouragement Rule)
\begin{enumerate}[label=(\roman*)]
\item \label{boundedgvarphi} 
$\sup_{x,z}|g(x,z)|<\infty$ and $\sup_{x,u}|\varphi(x,u;p)|<\infty$. 
\item \label{subGaussian} 
$\varepsilon=Y-\varphi(X,p(X,Z);p)$ is uniformly sub-Gaussian conditionally on $(X,Z)$ and has second moments uniformly bounded from below.
\item \label{boundedhatgvarphi} 
$\sup_{x,z}|\hat{g}^{(-k)}(x,z)|<\infty$ and  $\sup_{x,u}|\hat{\varphi}^{(-k)}(x,u;\hat{p}^{(-k)})|<\infty$ almost surely.
\item \label{L2error}
There exist $0<\zeta_g,\zeta_\varphi<1$ with $\zeta_g+\zeta_\varphi\geq1$ and $a(n)\to0$ such that
\begin{align*}
    E[(\hat{g}^{(-k(i))}(X_i,Z_i)-g(X_i,Z_i))^2]&\leq\frac{a(n)}{n^{\zeta_g}},\\
    E[(\hat{\varphi}^{(-k(i))}(X_i,\hat{p}^{(-k(i))}(X_i,Z_i);\hat{p}^{(-k(i))})-\varphi(X_i,p(X_i,Z_i);p))^2]&\leq\frac{a(n)}{n^{\zeta_\varphi}},
\end{align*}
and for $d\in\{0,1\}$,
\begin{equation*}
    E[(\hat{\varphi}^{(-k(i))}(X_i,\hat{p}^{(-k(i))}(X_i,\alpha_d(X_i,Z_i));\hat{p}^{(-k(i))})-\varphi(X_i,p(X_i,\alpha_d(X_i,Z_i));p))^2]\leq\frac{a(n)}{n^{\zeta_\varphi}},
\end{equation*}
\end{enumerate}
\end{enumerate}

Theorem \ref{upperDR} shows that the average regret of $\hat{\pi}_{\mathrm{DR}}$ decays no slower than $n^{-1/2}$.

\begin{theorem}\label{upperDR}
Suppose that Assumptions \ref{MTErestrictions}, \ref{identifyW}, \ref{boundedvc}(ii), and \ref{DRassumption} hold. Then,
\begin{equation*}
    E[R(\hat{\pi}_{\mathrm{DR}})]=O(n^{-1/2}).
\end{equation*}
\end{theorem}

\begin{proof}
    I follow \citet{athey2021policy} to work with
\begin{align*}
    A(\pi)&=E[(2\pi(X,Z)-1)\cdot(\varphi(X,p(X,\alpha_1(X,Z));p)-\varphi(X,p(Z,\alpha_0(X,Z));p))],\\
    \hat{A}_n(\pi)&=\frac{1}{n}\sum_{i=1}^n(2\pi(X_i,Z_i)-1)\hat{\Gamma}_i.
\end{align*}
Note that $\max_{\pi'\in\Pi}A(\pi')-A(\pi)=2R(\pi)$. It is convenient to define an ideal version of the objective in (\ref{DRobjective}) based on the true influence scores:
\begin{align*}
    \tilde{A}_n(\pi)&=\frac{1}{n}\sum_{i=1}^n(2\pi(X_i,Z_i)-1)\Gamma_i,\\
    \Gamma_i&=\varphi(X_i,p(X_i,\alpha_1(X_i,Z_i));p)-\varphi(X_i,p(X_i,\alpha_0(X_i,Z_i));p)\\
    &\quad\ +g(X_i,Z_i)(Y_i-\varphi(X_i,p(X_i,Z_i);p)).
\end{align*}
By writing
\begin{equation*}
    \hat{A}_n(\pi)-A(\pi)=\hat{A}_n(\pi)-\tilde{A}_n(\pi)+\tilde{A}_n(\pi)-A(\pi),
\end{equation*}
I study stochastic fluctuations of $\hat{A}_n(\pi)-A(\pi)$ for $\pi\in\Pi$ in two steps. First, I bound $|\tilde{A}_n(\pi)-A(\pi)|$ over $\lambda$-slices of $\Pi$ defined as
\begin{equation*}
    \Pi^\lambda=\{\pi\in\Pi:R(\pi)\leq\lambda\}.
\end{equation*}
By Assumptions \ref{DRassumption}(i)--(ii), the $\Gamma_i$ are uniformly sub-Gaussian and have variance bounded from below. Hence, Corollary 3 of \citet{athey2021policy} holds with $S_n=E[\Gamma^2]$ and $S_n^\lambda=\sup\{\Var[(2\pi(X,Z)-1)\Gamma]:\pi\in\Pi^\lambda\}$. Next, I bound $|\hat{A}_n(\pi)-\tilde{A}_n(\pi)|$ over $\pi\in\Pi$. To save space, write
\begin{align*}
    \Delta\varphi_i&=\varphi(X_i,p(X_i,\alpha_1(X_i,Z_i));p)-\varphi(X_i,p(X_i,\alpha_0(X_i,Z_i));p),\\
    \Delta\hat{\varphi}_i&=\hat{\varphi}^{(-k(i))}(X_i,\hat{p}^{(-k(i))}(X_i,\alpha_1(X_i,Z_i));\hat{p}^{(-k(i))})-\hat{\varphi}^{(-k(i))}(X_i,\hat{p}^{(-k(i))}(X_i,\alpha_0(X_i,Z_i));\hat{p}^{(-k(i))}).
\end{align*}
For any fixed $\pi$, one can expand $\hat{A}_n(\pi)-\tilde{A}_n(\pi)$ as
\begin{equation*}
    \hat{A}_n(\pi)-\tilde{A}_n(\pi)=D_1(\pi)+D_2(\pi)-D_3(\pi),
\end{equation*}
where
\begin{align*}
    D_1(\pi)&=\frac{1}{n}\sum_{i=1}^n(2\pi(X_i,Z_i)-1)(Y_i-\varphi(X_i,p(X_i,Z_i);p))(\hat{g}^{(-k(i))}(X_i,Z_i)-g(X_i,Z_i)),\\
    D_2(\pi)&=\frac{1}{n}\sum_{i=1}^n(2\pi(X_i,Z_i)-1)[\Delta\hat{\varphi}_i-\Delta\varphi_i\\
    &\quad\ -g(X_i,Z_i)(\hat{\varphi}^{(-k(i))}(X_i,\hat{p}^{(-k(i))}(X_i,Z_i);\hat{p}^{(-k(i))})-\varphi(X_i,p(X_i,Z_i);p))],\\
    D_3(\pi)&=\frac{1}{n}\sum_{i=1}^n(2\pi(X_i,Z_i)-1)(\hat{\varphi}^{(-k(i))}(X_i,\hat{p}^{(-k(i))}(X_i,Z_i);\hat{p}^{(-k(i))})-\varphi(X_i,p(X_i,Z_i);p))\\
    &\quad\ \times(\hat{g}^{(-k(i))}(X_i,Z_i)-g(X_i,Z_i)).
\end{align*}
I bound these three summands separately. I start with $D_1(\pi)$. It is helpful to separate out the contributions of the $K$ different folds:
\begin{equation*}
    D_1^{(k)}(\pi)=\frac{1}{n_k}\sum_{i:k(i)=k}(2\pi(X_i,Z_i)-1)(Y_i-\varphi(X_i,p(X_i,Z_i);p))(\hat{g}^{(-k(i))}(X_i,Z_i)-g(X_i,Z_i))
\end{equation*}
so that $D_1(\pi)=\sum_{k=1}^K \frac{n_k}{n} D_1^{(k)}(\pi)$, where $n_k=|\{i:k(i)=k\}|$ denotes the number of observations in the $k$th fold. Note that $E[Y_i-\varphi(X_i,p(X_i,Z_i);p)|X_i,Z_i,\hat{g}^{(-k(i))}(\cdot)]=0$. Hence, conditional on $\hat{g}^{(-k)}(\cdot)$ fit on the other $K-1$ folds, $D_1^{(k)}(\pi)$ has zero mean and asymptotic variance
\begin{equation*}
    V_1(k)=E[(\hat{g}^{(-k(i))}(X_i,Z_i)-g(X_i,Z_i))^2\operatorname{Var}[Y_i-\varphi(X_i,p(X_i,Z_i);p)|X_i,Z_i]|\hat{g}^{(-k)}(\cdot)].
\end{equation*}
By Assumptions \ref{DRassumption}(i)--(iii), the individual summands in $D_1^{(k)}(\pi)$ are uniformly sub-Gaussian almost surely. Then, one can apply Corollary 3 of \citet{athey2021policy} to establish that
\begin{equation}
    E\Big[\sup_{\pi\in\Pi}\big|D_1^{(k)}(\pi)\big||\hat{g}^{(-k)}(\cdot)\Big]=O\Big(\sqrt{\operatorname{VC}(\Pi)\frac{V_1(k)}{n_k}}\Big).
    \label{boundD1k}
\end{equation}
Since for a finite number of evenly-sized folds, $n_k/n\to K^{-1}$, one can use Assumption \ref{DRassumption}(iv) to check that
\begin{equation*}
    E[V_1(k)]=O\Big(\frac{a((1-K^{-1})n)}{n^{\zeta_g}}\Big).
\end{equation*}
Then, applying (\ref{boundD1k}) to all $K$ folds and using Jensen's inequality, one can find that
\begin{equation}
E\Big[\sup_{\pi\in\Pi}\big|D_1(\pi)\big|\Big]=O\Big(\sqrt{\operatorname{VC}(\Pi)\frac{a((1-K^{-1})n)}{n^{1+\zeta_g}}}\Big).
\label{boundD1}
\end{equation}
I proceed to bound $D_2(\pi)$. Since for any integrable function $\tilde{m}:\mathcal{X}\times\mathbb{R}\to\mathbb{R}$, $E[\tilde{m}(X,\alpha_1(X,Z))-\tilde{m}(X,\alpha_0(X,Z))-g(X,Z)\tilde{m}(X,Z)]=0$,
\begin{align*}
    E[\Delta\hat{\varphi}_i-\Delta\varphi_i-g(X_i,Z_i)(&\hat{\varphi}^{(-k(i))}(X_i,\hat{p}^{(-k(i))}(X_i,Z_i);\hat{p}^{(-k(i))})\\
    &-\varphi(X_i,p(X_i,Z_i);p))|\hat{p}^{(-k(i))}(\cdot),\hat{\varphi}^{(-k(i))}(\cdot;\hat{p}^{(-k(i))}(\cdot))]=0.
\end{align*}
Thus, by a similar argument as before, one can show that
\begin{equation}
    E\Big[\sup_{\pi\in\Pi}\big|D_2(\pi)\big|\Big]=O\Big(\sqrt{\operatorname{VC}(\Pi)\frac{a((1-K^{-1})n)}{n^{1+\zeta_\varphi}}}\Big).
    \label{boundD2}
\end{equation}
It remains to bound $D_3(\pi)$. By the Cauchy-Schwarz inequality,
\begin{align*}
    |D_3(\pi)|&\leq\sqrt{\frac{1}{n}\sum_{i=1}^n(\hat{\varphi}^{(-k(i))}(X_i,\hat{p}^{(-k(i))}(X_i,Z_i);\hat{p}^{(-k(i))})-\varphi(X_i,p(X_i,Z_i);p))^2}\\
    &\quad\ \times\sqrt{\frac{1}{n}\sum_{i=1}^n(\hat{g}^{(-k(i))}(X_i,Z_i)-g(X_i,Z_i))^2}.
\end{align*}
Since this bound does not depend on $\pi$, one can apply the Cauchy-Schwarz inequality again to obtain
\begin{align}
    E\Big[\sup_{\pi\in\Pi}|D_3(\pi)|\Big]&\leq\sqrt{E[(\hat{\varphi}^{(-k(i))}(X_i,\hat{p}^{(-k(i))}(X_i,Z_i);\hat{p}^{(-k(i))})-\varphi(X_i,p(X_i,Z_i);p))^2]}\nonumber\\
    &\quad\ \times\sqrt{E[(\hat{g}^{(-k(i))}(X_i,Z_i)-g(X_i,Z_i))^2]}\nonumber\\
    &=O\Big(\frac{a((1-K^{-1})n)}{\sqrt{n}}\Big).
    \label{boundD3}
\end{align}
Combining (\ref{boundD1})--(\ref{boundD3}), one has
\begin{equation*}
    \sqrt{n}E[\sup\{|\hat{A}_n(\pi)-\tilde{A}_n(\pi)|:\pi\in\Pi\}]=O\Big(1+\sqrt{\operatorname{VC}(\Pi)\frac{a((1-K^{-1})n)}{n^{\min\{\zeta_g,\zeta_\varphi\}}}}\Big),
\end{equation*}
which can be viewed as a counterpart of Lemma 4 of \citet{athey2021policy}. The desired result follows from the proof of Theorem 1 of \citet{athey2021policy}.
\end{proof}

\section{Proof of Auxiliary Lemmas}
\label{proofauxiliary}
\renewcommand{\theequation}{G.\arabic{equation}}
\setcounter{equation}{0}

\begin{proof}[Proof of Lemma \ref{BBnI}]
Setting $\Xi_i=n^{-1}(\tilde{b}^k(\tilde{X}_i)\tilde{b}^k(\tilde{X}_i)^\top-I_k)$ and noting that 
\begin{align*}
    \max_{1\leq i\leq n}\|\Xi_i\|&\leq n^{-1}(\zeta_k^2\lambda_k^2+1),\\
    \max\Big\{\Big\|\sum_{i=1}^n E[\Xi_i\Xi_i^\top]\Big\|,\Big\|\sum_{i=1}^n E[\Xi_i^\top\Xi_i]\Big\|\Big\}&\leq n^{-1}(\zeta_k^2\lambda_k^2+1),
\end{align*}
Then, the desired result follows from Theorem 4.1 of \citet{chen2015optimal}.
\end{proof}

\begin{proof}[Proof of Lemma \ref{varianceterm}]
Fix $\tilde{x}\in\tilde{\mathcal{X}}$ and $\delta\in(0,1]$. By rotational invariance, one has
\begin{equation*}
    \hat{p}(\tilde{x})-\tilde{p}(\tilde{x})=\tilde{b}^k(\tilde{x})^\top(\tilde{B}^\top\tilde{B}/n)^-\tilde{B}^\top e/n,
\end{equation*}
where $e=(\epsilon_1,\dots,\epsilon_n)^\top$. Define $G_{i,n}(\tilde{x})=\tilde{b}^k(\tilde{x})^\top(\tilde{B}^\top\tilde{B}/n)^-\tilde{b}^k(\tilde{X}_i)1_{\mathcal{A}_n}$. Then $(\hat{p}(\tilde{x})-\tilde{p}(\tilde{x}))1_{\mathcal{A}_n}=\frac{1}{n}\sum_{i=1}^n G_{i,n}(\tilde{x})\epsilon_i$. Note that $\|(\tilde{B}^\top\tilde{B}/n)^{-1}\|\leq2$ on $\mathcal{A}_n$. Then, for $q\geq 2$,
\begin{align*}
    &\quad \sum_{i=1}^n E[(n^{-1}G_{i,n}(\tilde{x}))^2|\epsilon_i|^q|\tilde{X}_1^n]\\
    &=\frac{1}{n^2}\sum_{i=1}^n E[|\epsilon_i|^q|\tilde{X}_i]\tilde{b}^k(\tilde{x})^\top(\tilde{B}^\top\tilde{B}/n)^-\tilde{b}^k(\tilde{X}_i)\tilde{b}^k(\tilde{X}_i)^\top1_{\mathcal{A}_n}(\tilde{B}^\top\tilde{B}/n)^-\tilde{b}^k(\tilde{x})\\
    &\leq\sup_{\tilde{x}}E[|\epsilon|^q|\tilde{X}=\tilde{x}]\frac{1}{n^2}\sum_{i=1}^n\tilde{b}^k(\tilde{x})^\top(\tilde{B}^\top\tilde{B}/n)^-\tilde{b}^k(\tilde{X}_i)\tilde{b}^k(\tilde{X}_i)^\top1_{\mathcal{A}_n}(\tilde{B}^\top\tilde{B}/n)^-\tilde{b}^k(\tilde{x})\\
    &\leq\frac{q!}{2}\nu^2\tilde{c}^{q-2}\frac{1}{n}\tilde{b}^k(\tilde{x})^\top(\tilde{B}^\top\tilde{B}/n)^-(\tilde{B}^\top\tilde{B}/n)1_{\mathcal{A}_n}(\tilde{B}^\top\tilde{B}/n)^-\tilde{b}^k(\tilde{x})\\
    &\leq\frac{q!}{2}\nu^2\tilde{c}^{q-2}\frac{2\zeta_k^2\lambda_k^2}{n}.
\end{align*}
Moreover, for $q\geq3$,
\begin{align*}
    &\quad \sum_{i=1}^n E[|n^{-1}G_{i,n}(\tilde{x})\epsilon_i|^q|\tilde{X}_1^n]\\
    &\leq n^{-(q-2)}\sup_{\tilde{x}}|\tilde{b}^k(\tilde{x})^\top(\tilde{B}^\top\tilde{B}/n)^-1_{\mathcal{A}_n}\tilde{b}^k(\tilde{x})|^{q-2}\sum_{i=1}^n E[(n^{-1}G_{i,n}(\tilde{x}))^2|\epsilon_i|^q|\tilde{X}_1^n]\\
    &\leq\Big(\frac{2\zeta_k^2\lambda_k^2}{n}\Big)^{q-2}\frac{q!}{2}\nu^2\tilde{c}^{q-2}\frac{2\zeta_k^2\lambda_k^2}{n}\\
    &=\frac{q!}{2}\nu^2\frac{2\zeta_k^2\lambda_k^2}{n}\Big(\tilde{c}\frac{2\zeta_k^2\lambda_k^2}{n}\Big)^{q-2}.
\end{align*}
Hence, the conditions of Theorem 2.10 of \citet{boucheron2013concentration} hold with $v=\nu^2\frac{2\zeta_k^2\lambda_k^2}{n}$ and $c=\tilde{c}\frac{2\zeta_k^2\lambda_k^2}{n}$. Since $v$ and $c$ do not depend on $\tilde{x}$, by their Corollary 2.11,
\begin{align*}
    \Pr(1_{\mathcal{A}_n}\|\hat{p}-\tilde{p}\|_\infty\geq\delta)
    &=\Pr\Big(\sup_{\tilde{x}}\Big|\frac{1}{n}\sum_{i=1}^n G_{i,n}(\tilde{x})\epsilon_i\Big|\geq \delta\Big)\\
    &\leq\exp\Big(-\frac{n\delta^2}{4\zeta_k^2\lambda_k^2(\nu^2+\tilde{c}\delta)}\Big)\\
    &\leq\exp\Big(-\frac{C_6n\delta^2}{\zeta_k^2\lambda_k^2}\Big)
\end{align*}
for some finite positive constant $C_6$.
\end{proof}

\begin{proof}[Proof of Lemma \ref{biasterm}]
First, for any $h\in B_k$,
\begin{align*}
    \|\tilde{p}-p\|_\infty&=\|\tilde{p}-h+h-p\|_\infty\\
    &=\|P_{k,n}(p-h)+h-p\|_\infty\\
    &\leq\|P_{k,n}(p-h)\|_\infty+\|p-h\|_\infty\\
    &\leq(1+\|P_{k,n}\|_\infty)\|p-h\|_\infty.
\end{align*}
Taking the infimum over $h\in B_k$ yields the first result. Second, take any $h\in L^\infty(\tilde{X})$ with $\|h\|_\infty\neq0$. By the Cauchy-Schwarz inequality, one has
\begin{align*}
    |1_{\mathcal{A}_n}P_{k,n}h(\tilde{x})|&\leq\|\tilde{b}^k(\tilde{x})\|\|1_{\mathcal{A}_n}(\tilde{B}^\top\tilde{B}/n)^-\tilde{B}^\top\mathbf{H}/n\|\\
    &\leq\zeta_k\lambda_k\|1_{\mathcal{A}_n}(\tilde{B}^\top\tilde{B}/n)^-\tilde{B}^\top\mathbf{H}/n\|
\end{align*}
uniformly over $\tilde{x}$. On $\mathcal{A}_n$, $\|\tilde{B}^\top\tilde{B}/n-I_k\|\leq\frac{1}{2}$ and thus $\lambda_{\mathrm{min}}(\tilde{B}^\top\tilde{B}/n)\geq\frac{1}{2}$. Then,
\begin{align*}
    \|1_{\mathcal{A}_n}(\tilde{B}^\top\tilde{B}/n)^-\tilde{B}^\top\mathbf{H}/n\|^2&=1_{\mathcal{A}_n}(\mathbf{H}^\top\tilde{B}/n(\tilde{B}^\top\tilde{B}/n)^{-1}(\tilde{B}^\top\tilde{B}/n)^{-1}\tilde{B}^\top\mathbf{H}/n\\
    &\leq2\mathbf{H}^\top\tilde{B}/n(\tilde{B}^\top\tilde{B}/n)^{-1}\tilde{B}^\top\mathbf{H}/n\\
    &\leq2\|h\|^2\\
    &\leq2\|h\|_\infty^2,
\end{align*}
where the second-last inequality follows from $1_{\mathcal{A}_n}\tilde{B}(\tilde{B}^\top\tilde{B})^{-1}\tilde{B}^\top$ being idempotent. It follows that $\|1_{\mathcal{A}_n}P_{k,n}h\|_\infty/\|h\|_\infty\leq \sqrt{2}\zeta_k\lambda_k$ uniformly in $h$. Taking the sup over $h$ yields the second result.
\end{proof}

\begin{proof}[Proof of Lemma \ref{influence}]
For any $p$ and $\varphi$, one has
\begin{equation*}
    E[m(X,Z;p,\varphi)]=E[\pi(X,Z)g(X,Z)\varphi(X,p(X,Z);p)].
\end{equation*}
The chain rule gives
\begin{align}
    \frac{\partial}{\partial\tau}E[m(X,Z;p_\tau,\varphi_\tau)]&=\frac{\partial}{\partial\tau}E[\pi(X,Z)g(X,Z)\varphi_\tau(X,p_0(X,Z);p_0)]\nonumber\\
    &\quad\ +\frac{\partial}{\partial\tau}E[\pi(X,Z)g(X,Z)\varphi_0(X,p_\tau(X,Z);p_0)]\nonumber\\
    &\quad\ +\frac{\partial}{\partial\tau}E[\pi(X,Z)g(X,Z)\varphi_0(X,p_0(X,Z);p_\tau)].
    \label{chainrule}
\end{align}
For the first term on the right-hand side of (\ref{chainrule}), note that
\begin{align*}
    \frac{\partial}{\partial\tau}E[\pi(X,Z)g(X,Z)\varphi_\tau(X,p_0(X,Z);p_0)]&=\frac{\partial}{\partial\tau}E_\tau[\pi(X,Z)g(X,Z)\varphi_\tau(X,p_0(X,Z);p_0)]\\
    &\quad\ -\frac{\partial}{\partial\tau}E_\tau[\pi(X,Z)g(X,Z)\varphi_0(X,p_0(X,Z);p_0)]\\
    &=\frac{\partial}{\partial\tau}E_\tau[\pi(X,Z)g(X,Z)(Y-\varphi_0(X,p_0(X,Z);p_0))]\\
    &=E_H[\pi(X,Z)g(X,Z)(Y-\varphi_0(X,p_0(X,Z);p_0))].
\end{align*}
The second and third terms on the right-hand side of (\ref{chainrule}) reflect the two contributions of the estimation error of $p$. Note that for any $\tau$,
\begin{equation*}
    E[\pi(X,Z)g(X,Z)(Y-\varphi_0(X,p_\tau(X,Z);p_\tau))]=0.
\end{equation*}
Differentiating with respect to $\tau$ and evaluating the result at $\tau=0$, one can find that
\begin{equation*}
    \frac{\partial}{\partial\tau}E[\pi(X,Z)g(X,Z)\varphi_0(X,p_\tau(X,Z);p_0)]+\frac{\partial}{\partial\tau}E[\pi(X,Z)g(X,Z)\varphi_0(X,p_0(X,Z);p_\tau)]=0.
\end{equation*}
Putting everything together yields the desired result.
\end{proof}

\section{Additional Tables and Figures}
\label{addtab}
\renewcommand{\thetable}{H.\arabic{table}}
\renewcommand{\thefigure}{H.\arabic{figure}}
\setcounter{table}{0}
\setcounter{figure}{0}

\renewcommand{\arraystretch}{0.64}
\begin{table}[H]
    \centering
    \caption{Sample Averages for the Treatment and Control Groups}
    \begin{tabular}{l c c}
    \hline
    & Upper secondary or higher  & Less than upper secondary\\
    & (treatment group) & (control group)\\
    & $N=841$ & $N=1263$ \\
    \hline\\
    Log hourly wages (rupiah) &  8.018 & 7.209\\
    Years of education & 13.128 & 5.585\\
    Distance to school (km) & 1.529 & 1.565\\
    Distance to health post (km) & 0.331 & 0.361\\
    Fees per continuing student & 3.464 & 3.992\\
    \hspace{0.3cm}(1000 rupiah) & & \\
    Age & 35.668 & 36.766\\
    Religion Protestant & 0.037 & 0.008\\
    \hspace{0.3cm} Catholic & 0.023 & 0.007\\
    \hspace{0.3cm} Other & 0.075 & 0.039\\
    \hspace{0.3cm} Muslim & 0.866 & 0.946\\
    Father uneducated &  0.189 & 0.254\\
    \hspace{0.3cm} elementary & 0.325 & 0.251\\
    \hspace{0.3cm} secondary and higher & 0.275& 0.026\\
    \hspace{0.3cm} missing & 0.212 & 0.468\\
    Mother uneducated & 0.182 & 0.203\\
    \hspace{0.3cm} elementary & 0.301 & 0.173\\
    \hspace{0.3cm} secondary and higher & 0.138 & 0.011\\
    \hspace{0.3cm} missing & 0.379 & 0.612\\
    Rural household & 0.483 & 0.644\\
    North Sumatra & 0.034 & 0.045\\
    West Sumatra & 0.023 & 0.023\\
    South Sumatra & 0.069 & 0.033\\
    Lampung & 0.013 & 0.028\\
    Jakarta & 0 & 0\\
    Central Java & 0.102 & 0.216\\
    Yogyakarta & 0.121 & 0.077\\
    East Java & 0.152 & 0.201\\
    Bali & 0.081 & 0.035\\
    West Nussa Tengara & 0.084 & 0.053\\
    South Kalimanthan & 0.040 & 0.033\\
    South Sulawesi & 0.040 & 0.019\\
    \hline
    \end{tabular}
    \label{samplestat}
\end{table}
\end{document}